\def \VersionAuthor {}
\ifdefined\VersionAuthor
	\newcommand{\AuthorVersion}[1]{#1}
	\newcommand{\SpringerVersion}[1]{}
\else
	\newcommand{\AuthorVersion}[1]{}
	\newcommand{\SpringerVersion}[1]{#1}
\fi

\documentclass[a4paper,10pt]{llncs}
\usepackage[utf8]{inputenc}

\usepackage[ruled,vlined,linesnumbered]{algorithm2e}
	\SetKwInOut{Input}{input}
	\SetKwInOut{Output}{output}

\usepackage{subcaption}

\usepackage[svgnames,table]{xcolor}
\definecolor{darkblue}{rgb}{0.0,0.0,0.6}
\definecolor{darkgreen}{rgb}{0, 0.5, 0}
\definecolor{darkpurple}{rgb}{0.7, 0, 0.7}
\definecolor{darkblue}{rgb}{0, 0, 0.7}

\usepackage[
		colorlinks=true,
	\ifdefined \VersionWithComments
		pagebackref=true,
	\fi
		citecolor=darkgreen,
		linkcolor=darkblue,
		urlcolor=darkpurple,
	]{hyperref}

\usepackage{amsmath} % for \text et al., for ``cases'' environment…
\usepackage{amssymb} % for \mathbb et al.
\usepackage[misc,geometry]{ifsym} % For \Letter

\usepackage[svgnames,table]{xcolor}
\usepackage{tikz}
\usetikzlibrary{arrows,automata}
\ifdefined \VersionAuthor
	\newcommand{\LongVersion}[1]{\ifdefined\VersionWithComments{\color{red!40!black}#1}\else#1\fi}
	\newcommand{\ShortVersion}[1]{\ifdefined\VersionWithComments{\color{black!40}#1}\fi}
\else
	\newcommand{\LongVersion}[1]{\ifdefined\VersionWithComments{\color{black!40}#1}\fi}
	\newcommand{\ShortVersion}[1]{\ifdefined\VersionWithComments{\color{red!40!black}#1}\else#1\fi}
\fi

\ifdefined \VersionWithComments
 	\definecolor{colorok}{RGB}{255,0,255}
\else
	\definecolor{colorok}{RGB}{0,0,0}
\fi

\newcommand{\ie}{\textcolor{colorok}{i.\,e.,}}% \xspace
\newcommand{\vect}{\ensuremath{\textcolor{colorok}{\mathbb{V}}}}

\newcommand{\bounds}{\ensuremath{\textcolor{colorok}{\mathit{bounds}}}}
\newcommand{\boundinf}{\ensuremath{\textcolor{colorok}{\mathit{inf}}}}
\newcommand{\boundsup}{\ensuremath{\textcolor{colorok}{\mathit{sup}}}}

\newcommand{\action}{\ensuremath{\textcolor{colorok}{a}}}
\newcommand{\Alphabet}{\ensuremath{\textcolor{colorok}{\Sigma}}}
\newcommand{\clock}{\ensuremath{\textcolor{colorok}{x}}} % clock
\newcommand{\Clock}{\ensuremath{\textcolor{colorok}{\mathbb{X}}}} % set of clocks
\newcommand{\clockval}{\ensuremath{\textcolor{colorok}{v}}}
\newcommand{\clockvalzero}{\ensuremath{\textcolor{colorok}{\vec{0}}}}
\newcommand{\ClockVals}{\ensuremath{\textcolor{colorok}{\mathcal{V}}}}
\newcommand{\compOp}{\ensuremath{\color{colorok}\bowtie}}
\newcommand{\config}{\ensuremath{\textcolor{colorok}{\gamma}}}

\newcommand{\edge}{\ensuremath{\textcolor{colorok}{\delta}}}
\newcommand{\Edges}{\ensuremath{\textcolor{colorok}{\Delta}}}
\newcommand{\guard}{\ensuremath{\textcolor{colorok}{g}}}
\newcommand{\GuardsXP}{\ensuremath{\textcolor{colorok}{\mathcal{G}(\Clock,\Param)}}}
\newcommand{\interval}{\ensuremath{\textcolor{colorok}{\mathcal{I}}}}
\newcommand{\IntervalsQ}{\ensuremath{\textcolor{colorok}{\mathbb{I}_{\grandqplus}}}}
\newcommand{\loc}{\ensuremath{\textcolor{colorok}{q}}} % location
\newcommand{\locinit}{\ensuremath{\textcolor{colorok}{q_0}}} % initial location
\newcommand{\locfinal}{\ensuremath{\textcolor{colorok}{q_f}}} % final location
\newcommand{\Loc}{\ensuremath{\textcolor{colorok}{Q}}} % set of locations
\newcommand{\Param}{\ensuremath{\textcolor{colorok}{\mathbb{P}}}} % set of parameters (P)
\newcommand{\param}{\ensuremath{\textcolor{colorok}{\lambda}}} % parameter (p)
\newcommand{\paraml}{\textcolor{colorok}{\lambda^l}} % 
\newcommand{\paramu}{\textcolor{colorok}{\lambda^u}} % 
\newcommand{\pval}{\ensuremath{\textcolor{colorok}{p}}} % parameter valuation % WARNING: conflit p

\newcommand{\PTBP}{\ensuremath{\textcolor{colorok}{\mathcal{N}}}}
\newcommand{\update}{\ensuremath{\textcolor{colorok}{\mathit{up}}}}
\newcommand{\Updates}{\ensuremath{\textcolor{colorok}{\mathrm{Updates}}}}

\newcommand{\styleTwoCM}[1]{\ensuremath{\textcolor{colorok}{\mathbf{#1}}}}
\newcommand{\TwoCMconfig}{\styleTwoCM{s}}
\newcommand{\TwoCMCounter}{\styleTwoCM{C}}
\newcommand{\TwoCMCounterValue}{\styleTwoCM{c}}
\newcommand{\TwoCMInstructions}{\styleTwoCM{K}}
\newcommand{\TwoCMinstruction}{\styleTwoCM{k}}
\newcommand{\TwoCMinstructionZero}{\TwoCMinstruction_0}
\newcommand{\TwoCMinstructionAcc}{\TwoCMinstruction_{\mathit{acc}}}
\newcommand{\TWOCM}{\styleTwoCM{M}}

\newcommand{\boundinfin}[2]{\ensuremath{\textcolor{colorok}{\boundinf}(#1,#2)}}
\newcommand{\boundsupin}[2]{\ensuremath{\textcolor{colorok}{\boundsup}(#1,#2)}}

\newcommand{\dec}[1]{{\color{green!50!black}#1}}
\newcommand{\undec}[1]{{\color{red!70!black}#1}}
\newcommand{\celldec}[1]{{\cellcolor{green!50!white}\bfseries{}#1}}
\newcommand{\cellundec}[1]{{\cellcolor{red!50!white}\em{}#1}}
\newcommand{\cellopen}[1]{{\cellcolor{yellow!50!white}\em{}#1}}

\newcommand{\grandn}{\ensuremath{\textcolor{colorok}{\mathbb N}}}
\newcommand{\grandq}{\ensuremath{\textcolor{colorok}{\mathbb Q}}}
\newcommand{\grandqplus}{\ensuremath{\textcolor{colorok}{\grandq_{+}}}} % \geq 0
\newcommand{\grandr}{\ensuremath{\textcolor{colorok}{\mathbb R}}}
\newcommand{\grandrplus}{\ensuremath{\textcolor{colorok}{\grandr_{+}}}} % \geq 0

\newcommand{\state}{\ensuremath{state}}
\newcommand{\val}{\ensuremath{val}}

\newcommand{\pid}{\ensuremath{i}}
 
\newcommand{\Execs}{\mathcal{E}}

\newcommand{\client}{c}
\newcommand{\factory}{f}
\newcommand{\goal}{g}
\newcommand{\produ}{p}

\newcommand{\recallResult}[2]
{%	
	\smallskip

	\noindent\fcolorbox{black}{green!15}{
		\begin{minipage}{.95\columnwidth}
			\noindent\textbf{\cref{#1} (recalled).}
			{\em{}#2}
		\end{minipage}
	}
}

\usepackage[capitalise,english,nameinlink]{cleveref} % load after algorithm2e and hyperref

\ifdefined \VersionWithComments
	\usepackage{marginnote}
	\newcommand{\marginX}{\marginnote{\huge{\quad\quad\textbf{!}\quad\quad}}}
	\newcommand{\instructions}[1]{{\color{red}\marginX{}\textbf{Instructions: #1}}}
	\newcommand{\ea}[1]{{\color{blue}\marginX{}\textbf{[\'Etienne}: #1]}}
	\newcommand{\bd}[1]{{\color{red!50!black}\marginX{}\textbf{[BenoÃ®t}: #1]}}
	\newcommand{\dl}[1]{{\color{green!50!black}\marginX{}\textbf{[Didier}: #1]}}
	\newcommand{\nb}[1]{{\color{purple!90!black}\marginX{}\textbf{[Nathalie}: #1]}}
	\newcommand{\pf}[1]{{\color{orange!90!black}\marginX{}\textbf{[Paulin}: #1]}}
	\newcommand{\reviewer}[2]{\mbox{}{\color{red}\marginX{}\textbf{[Reviewer #1}: ``#2'']}}
	\newcommand{\todo}[1]{\mbox{}{\color{red}{\marginX{}\textbf{TODO}\ifx#1\\\else:\ \fi #1}}} % here, ``\\'' stands 	
	\newcommand{\toutfaux}[1]{}
\else
	\newcommand{\instructions}[1]{}
	\newcommand{\bd}[1]{}
	\newcommand{\ea}[1]{}
	\newcommand{\dl}[1]{}
	\newcommand{\nb}[1]{}
	\newcommand{\pf}[1]{}
	\newcommand{\reviewer}[2]{}
	\newcommand{\todo}[1]{}
	\newcommand{\toutfaux}[1]{}
\fi

\title{Parametric Timed Broadcast Protocols\thanks{%
	\LongVersion{This is the author (and extended) version of the manuscript of the same name published in the proceedings of the 20th International Conference on Verification, Model Checking, and Abstract Interpretation (VMCAI 2019).
	The final version is available at \url{http://dx.doi.org/10.1007/978-3-030-11245-5_23}.
	This version contains additional examples and all proofs, and fixes a typo in the name of the problems considered.}
	This work is partially supported by the ANR national research program PACS (ANR-14-CE28-0002)
	and
	by ERATO HASUO Metamathematics for Systems Design Project (No.\ JPMJER1603), JST.
}}
\author{\'Etienne Andr\'e\inst{1,2,3}\orcidID{0000-0001-8473-9555}\Letter
	\and
	Benoit Delahaye\inst{4}\orcidID{0000-0002-9104-4361}
	\and
	Paulin Fournier\inst{4}
	\and
	Didier Lime\inst{5}\orcidID{0000-0001-9429-7586}
}
\institute{Université Paris 13, LIPN, CNRS, UMR 7030, F-93430, Villetaneuse, France
\and
JFLI, CNRS, Tokyo, Japan
\email{eandre93430@lipn13.fr}
\and
National Institute of Informatics, Tokyo, Japan
\and
Université de Nantes, LS2N UMR CNRS 6004, Nantes, France
\and
\'{E}cole Centrale de Nantes, LS2N UMR CNRS 6004, Nantes, France
}

\begin{document}

% For all page numbers, except p.1
\AuthorVersion{
	\pagestyle{plain}
}

\maketitle

\ifdefined \VersionWithComments
	\textcolor{red}{\textbf{This is the version with comments. To disable comments, comment out line~3 in the \LaTeX{} source.}}
\fi

\instructions{VMCAI 2019: 
Submissions are restricted to 20 pages in Springer’s LNCS format, not counting references. Additional material may be placed in an appendix, to be read at the discretion of the reviewers and to be omitted in the final version. Formatting style files and further guidelines for formatting can be found at the Springer website.}

%\ea{je ne suis pas très fan du titre parce que la moitié du papier concerne AF qui n'est pas trop de l'accessibilité, enfin un peu moins.
%Pourquoi pas ``Reachability and unavoidability in PTBP ?'' Ou plus simplement ``Decision problems for PTBP'' ? Ou encore plus sobre ``Parametric timed broadcast protocols''? (qui a le mérite d'être un peu plus punchy aussi)
%}

\begin{abstract}
	In this paper we consider state reachability in networks composed of many identical processes running a parametric timed broadcast protocol (PTBP).
	PTBP are a new model extending both broadcast protocols and parametric timed automata.
	This work is, up to our knowledge, the first to consider the combination of both a parametric network size and timing parameters in clock guard constraints.
	Since the communication topology is of utmost importance in broadcast protocols, we investigate reachability problems in both clique semantics where every message reaches every processes, and in reconfigurable semantics where the set of receivers is chosen non-deterministically.
	In addition, we investigate the decidability status depending on whether the timing parameters in guards appear only as upper bounds in guards\LongVersion{ (U-PTBP)}, as lower bounds\LongVersion{ (L-PTBP)} or when the set of parameters is partitioned in lower-bound and upper-bound parameters\LongVersion{ (L/U-PTBP)}.\ea{a reviewer FSTTCS asked to remove these notations from the abstract}
\end{abstract}

\LongVersion{
\keywords{Parameterized systems, parametric timed model checking}
}

\ea{Il faut se mettre d'accord si au pluriel on écrit PTAs / PTBPs ou PTA/PTBP. Je n'ai aucune préférence personnellement. EDIT: je crois que Paulin a implicitement choisi de ne pas mettre de ``s''. OK.}

\ea{Un(e) reviewer de CONCUR 2018 demande à ce qu'on se compare à l'article suivant :
Laura Bozzelli, Sophie Pinchinat:
Verification of gap-order constraint abstractions of counter systems. Theor. Comput. Sci. 523: 1-36 (2014)
}

%%%%%%%%%%%%%%%%%%%%%%%%%%%%%%%%%%%%%%%%%%%%%%%%%%%%%%%%%%%%
%%%%%%%%%%%%%%%%%%%%%%%%%%%%%%%%%%%%%%%%%%%%%%%%%%%%%%%%%%%%
\section{Introduction}
%%%%%%%%%%%%%%%%%%%%%%%%%%%%%%%%%%%%%%%%%%%%%%%%%%%%%%%%%%%%
%%%%%%%%%%%%%%%%%%%%%%%%%%%%%%%%%%%%%%%%%%%%%%%%%%%%%%%%%%%%

%In order to deal with the increasing complexity and criticality of current computer systems, it has become necessary to analyze their behaviors, by making use of mathematically grounded models and techniques.
%Model-checking in particular provides a widely acknowledge framework for such analyses.\ea{pas convaincu que ces 2 premières phrases soient nécessaires à Fossacs}
The application of model-checking to real-life complex systems faces several problems, and for many of them the use of parameters, \ie{} symbolic constants representing an unknown quantity can be part of the solution.
First, for big systems, the so-called \emph{state-space explosion} limits the practical applicability of model-checking. Such big systems however are in general specified as the composition of smaller systems. A particularly interesting setting is the one in which all the components are identical, such as in many communication protocols. The number of involved components can then be abstracted away as a \emph{parameter}, with the hope of both overcoming the state-space explosion, and obtaining more useful answers from the model-checking process, such as ``for which sizes of the system does some property hold?''.
%
%% Second, the earlier in the development phase verification can be applied, the less costly will be fixing the problems found, but also the less information we have on the final system, in particular on many timing features, such as transmission times, watchdogs, etc.
%% Parameters can also be useful here by abstracting away the precise values of some yet unknown features, and at the same time allowing their dimensioning.
Second, the earlier in the development phase verification can be applied, the less costly will fixing the problems be. On the other hand, the earlier the verification is applied, the less information we have on the final system, in particular on many timing features, such as transmission times, watchdogs, etc.
Parameters can also be useful here by abstracting away the precise values of some yet unknown features, and at the same time allowing their dimensioning.

In this paper, we propose to combine two different types of parameters, namely the \emph{number of identical processes} and the \emph{timing features}, and study the decidability of classic parametric decision problems in the resulting formalism.
Both types of parameters, when introduced separately in timed automata-based formalisms, result in hard problems undecidable even in restricted settings.
% 	\dl{légèrement changé cette phrase}

Timed automata~\cite{AD94} extend finite-state automata with clocks, \ie{} real-valued variables that can be compared to constants in guards, and reset along transitions.
Parametric timed automata (PTA)~\cite{AHV93} allow to replace constants with unknown parameters in timing constraints.
The most basic verification question, ``does there exist a value for the parameters such that some location is reachable'' is undecidable with as few as 1 integer- or rational-valued parameter~\cite{Miller00,BBLS15}, or when only 1 clock is compared to a unique parameter~\cite{Miller00} (with additional clocks); see \cite{Andre18STTT} for a survey.
% 	\ea{ajouté BBLS15 qui améliore le résultat (on peut enlever si on manque de place)}\dl{OK}\ea{rajouté ``with additional clocks'' parce que sinon les non-experts comprennent souvent que 1c + 1p est indécidalble (j'ai souvent eu la remarque/question)}\ea{et en fait on ne manque pas de place puisque les réfs ne sont pas incluses dans les 12 pages de ICALP}
% 4 clocks and one parameter, with only one of the clock compared to the parameter~\cite{Miller00}.
The main syntactic subclass of PTA for which decidability is obtained is L/U-PTA~\cite{HRSV02}, in which the set of parameters is partitioned into lower-bound parameters (\ie{} parameters always compared as a lower bound in a clock guard) and upper-bound parameters (always as upper bounds). L/U-PTA have been shown~\cite{HRSV02} to be expressive enough to model classical examples from the literature, such as root contention or Fischer's mutual exclusion algorithm for instance.

Broadcast protocol networks~\cite{DSTZ-fsttcs12,delzanno2011power,delzanno2011parameterized,delzanno2010parameterized}, allow treating the size of a network as an unknown parameter. Here also the most simple basic verification question ``does there exist a value for the parameter such that some location is reachable by a process'' is undecidable when considering arbitrary communication topologies~\cite{delzanno2010parameterized}. However one can regain decidability by considering different communication topology settings. One option is to limit the topologies to cliques (every process receives every message)~\cite{delzanno2011power,delzanno2011parameterized,delzanno2010parameterized}. Another is to consider reconfigurable broadcasts in which the set of receivers is chosen non-deterministically at each step~\cite{DSTZ-fsttcs12}.
A timed version of this broadcast protocol was studied in~\cite{ADRST16}.
In the clique topology for this network, the reachability problem is decidable only when there is a single clock per process.

%\dl{il faut sûrement compléter}.

%\todo{citer \cite{DSTZ-fsttcs12,delzanno2011power,delzanno2011parameterized} puis surtout \cite{ADRST11} (ou \cite{ADRST16} plutôt) qui est le papier PTN le plus proche}

\paragraph*{Contributions}
In this work, we provide one more level of abstraction to the formalisms of the literature by proposing \emph{parametric timed broadcast protocols} (PTBP), \ie{} a new formalism made of an arbitrary number of identical timed processes in which timing parameters can be used.
A combination of two kinds of parameters seems natural, for example when designing and verifying communication protocols. Indeed, those protocols are required to work independently of the number of participants (hence the parametric size of networks) and the time constraints in each process are of paramount importance and thus could be tweaked in early development thanks to timing parameters.
% 	\ea{ici: rajouter une petite phrase qui dit que cette abstraction était celle que tout le monde attendait, va changer le monde, faire revenir l'être aimé, guérir le SIDA et justifier notre best paper à ICALP ?}\pf{done. Voir si ca vous va. J'ai rien trouvé sur le SIDA mais si je trouve j'hésite pas ...}\ea{parfait ! pour le SIDA, je propose un nombre inconnu de personnes qui… Non j'arrête.}%
This work is, up to our knowledge, the first to consider the combination of both a parametric network size and timing parameters in clock guard constraints.
We consider the following problems: does there exist a number of processes for which the set of timing parameter valuations allowing to reach a given location for one run (``EF''), or for all runs (``AF'') is non-empty (or universal)?
This gives rise to 4 problems: EF-existence, EF-universality, AF-existence and AF-universality.
% A natural question is whether there exists a number of processes and some valuations of the timing parameters for which a given location is reachable, either for some run (``EF''), or for all runs (``AF'').
As PTBP can be seen as an extension of both broadcast protocols and parametric timed automata, undecidability follows immediately from the existing undecidability results known for these two formalisms.
However, combining decidable subclasses of both formalisms is challenging, and does not necessarily make the EF and AF problems decidable for PTBP.

The communication topology is of utmost importance in broadcast protocols, and we therefore investigate reachability problems depending on the broadcast semantics.
In the reconfigurable semantics (where the set of receivers is chosen non-deterministically), AF-existence and AF-universality are decidable for 1-clock PTBP, and undecidable from 3~clocks even for L/U-PTBP with the same parameters partitioning as in L/U-PTA (the 2-clock case is equivalent\ea{en fait, on l'a pas montré} to a well-known open problem for PTA).\bd{cette phrase est moche}
The AF results may not seem surprising, as they resemble equivalent results for PTA.
However, EF-existence and EF-universality becomes undecidable even for 1-clock PTBP: this result comes in contrast with both non-parametric timed broadcast protocols and PTA for which the 1-clock case is decidable.

In the clique semantics (where every message reaches every process), we show that AF problems are undecidable even without any clock.
Then, as it is known that 2 clocks (and no parameter) yield undecidability, we study EF problems over 1 clock.
We investigate the decidability status depending on whether the timing parameters in guards appear only as upper bounds in guards (U-PTBP), as lower bounds (L-PTBP) or when the set of parameters is partitioned in lower-bound and upper-bound parameters (L/U-PTBP).
We show that L/U-PTBP become decidable for EF-existence (but not universality) when the parameter domain is bounded.
For EF-universality, decidability is obtained only for L-PTBP and U-PTBP for a parameter domain bounded with closed bounds.
%
% BEGIN À REMETTRE EN VERSION FINALE
The decidability border between L/U-PTA with a bounded parameter domain with closed bounds, and L/U-PTA with closed bounds was already spotted in~\cite{ALime17}, for liveness properties.
Our contributions are summarized in \cref{table:summary} (page~\pageref{table:summary}).

% \todo{By combining the known decidable subclasses, we are able to find some decidable subclasses of our doubly parameterized formalism\dots}

\paragraph*{Related work}
The concept of identical processes has been addressed in various settings, such as regular model checking~\cite{BJNT00}, or network of identical timed processes~\cite{AJ03,ADM04,ADRST11}.

To the best of our knowledge, combining two types of parameters (\ie{} discrete and continuous) was very little studied---with a few exceptions.
In \cite{DKRT97}, an attempt is made to mix discrete and continuous timing parameters (in an even non-linear fashion, \ie{} where parameters can be multiplied by other parameters).
However, the approach is fully \emph{ad-hoc} and addresses an extension of PTA, for which problems are already undecidable.
% 	\dl{et pourquoi c'est pas bien?}\ea{j'ai complété}
In \cite{LSLD15,delzanno2004automatic}, security protocols are studied with unknown timing constants, and an unbounded number of participants.
	However, the focus is not on decidability, and the general setting is undecidable.\ea{note pour moi : rechecker}
In~\cite{AKPP16}, action parameters (that can be seen as Booleans) and continuous timing parameters are combined (only linearly though) in an extension of PTA; the mere emptiness of the sets of action and timing parameters for which a location is reachable is undecidable.
In contrast, we exhibit in this work some decidable cases.

\paragraph*{Outline}
We introduce necessary definitions in \cref{section:definitions}.
We then study the existence and the universality problems for which a state is reachable and unavoidable respectively, in reconfigurable semantics (\cref{section:reconfigurable}) and clique semantics (\cref{section:clique}). We then investigate a restriction of the protocols, namely the L/U restriction (\cref{section:LU}).
%\todo{L/U}
We conclude in \cref{section:conclusion}.

%%%%%%%%%%%%%%%%%%%%%%%%%%%%%%%%%%%%%%%%%%%%%%%%%%%%%%%%%%%%
%%%%%%%%%%%%%%%%%%%%%%%%%%%%%%%%%%%%%%%%%%%%%%%%%%%%%%%%%%%%
\section{Definitions}\label{section:definitions}
%%%%%%%%%%%%%%%%%%%%%%%%%%%%%%%%%%%%%%%%%%%%%%%%%%%%%%%%%%%%
%%%%%%%%%%%%%%%%%%%%%%%%%%%%%%%%%%%%%%%%%%%%%%%%%%%%%%%%%%%%

\LongVersion{
\subsection{Notations}
}
We denote by $\grandn$, $\grandqplus$, and $\grandrplus$ the sets of all natural, non-negative rational, and non-negative real numbers respectively.%\pf{j'ai rajouté positive}\ea{sauf erreur c'est non-negative}
$[a,b]$ denotes the interval containing all rational numbers $\clock$ such that $\clock\leq b$ and $\clock\geq a$. As usual, we write $(a,b]$ to exclude $a$ from this set and $[a,b)$ to exclude $b$ (in which case we allow $b=+\infty$).%\ea{j'ai changé pour les notations anglophones : $[a,b[$ s'écrit $[a,b)$ en anglais}
We denote by $\IntervalsQ$ the set of all rational intervals.\ea{a priori utilisé une seule fois dans tout le papier…}

%Given a set $E$, we denote by $\mult(E)$ the set of all multi-sets of $E$ \ie{} the set of functions of $E\to \grandn$. Given a multi-set $m\in\mult(E)$ and an element $e\in E$, we denote by $m+e$ the multi-set such that $(m+e)(e)=m(e)+1$ and $(m+e)(e')=m(e')$ for $e'\neq e$. If $m(e)>0$ we define similarly $m-e$. We extend these definitions to additions and subtraction of multi-sets in a straightforward way.

Given a set $E$, and an integer $n\in\grandn$ we denote $\vect_n(E)$ the set of all vectors composed by $n$ elements of $E$.
We denote $\vect(E)$ the set of all vectors \ie{} $\vect(E)=\cup_{n\in\grandn}\vect_n(E)$.

Given a set of clocks $\Clock$, a valuation of $\Clock$ is a function of $\Clock\to\grandrplus$.
We denote by $\ClockVals(\Clock)$ the set of all valuations of $\Clock$ or just $\ClockVals$ when $\Clock$ is clear from the context.\ea{en fait, c'est systématiquement ce dernier cas, dans le papier}
The valuation assigning $0$ to all clock is written $\clockvalzero$.
Given a valuation $\clockval\in\ClockVals$ and a real number $t$ we denote by $\clockval+t$ the valuation $\clockval'$ such that for all $\clock\in \Clock$, $\clockval'(\clock)=\clockval(\clock)+t$, and $\clockval-t$ (if it exists%\dl{dans la mesure où c'est $\grandrplus$ au dessus, il n'y a pas de raison qu'elle n'existe pas? Mais classiquement c'est vrai qu'on se restreint plutôt aux réels positifs ou nuls}\ea{ça a dû changer mais au-dessus c'est bien positif ou nul donc la remarque de Paulin est nécessaire}
) the valuation such that $(\clockval-t)+t=\clockval$.
Given a set of clocks $\Clock$ and a set of parameters $\Param$ we write $\GuardsXP$ for the set of all sets of constraints of the form $\clock \compOp a$ with $\clock\in \Clock$, ${\compOp} \in \{<,\leq,=,\geq,>\}$ and $a\in \grandqplus\cup \Param$.\bd{On autorise pas les conjonctions/disjonctions? -- OK, compris, ``sets'' of constraints.}

We denote by $\Updates(\Clock)$ the set of updates of the clocks, where an update is a function $up:\ClockVals\to\ClockVals$ such that for all $\clock\in \Clock$, either $\update(\clockval)(\clock)=\clockval(\clock)$ or $\update(\clockval)(\clock)=0$.
When convenient we represent the update function with the set $\{\clock_1,\dots,\clock_k\}$ %\ea{pourquoi pas juste $\{\clock_1,\dots,\clock_k\}$ ?} 
 representing that clocks $\clock_1$ to $\clock_k$ are reset to $0$ while other clocks (here $\clock_i$ with $i>k$) are left unchanged.

Given a clock valuation $\clockval\in \Clock\to\grandrplus$ and a rational valuation of the variables $\pval:\Param \to\grandqplus$  we say that the valuation $\clockval$ satisfies a guard $\guard \in\GuardsXP$, written $\clockval\models_{\pval} \guard$ if for all $\clock\compOp a\in \guard$ %\ea{non défini} 
either $a\in \grandqplus$ and $\clockval(\clock)\compOp a$ or $a\in \Param$ and $\clockval(\clock)\compOp \pval(a)$.

%\ea{les paramètres sont pas définis, là…}

%Given a set of parameters $\Param$%\ea{ben non, pas variables, mais paramètres}
%, a variable bound is a function $b:P\to \IntervalsQ$ that restrict the valuation of the variable.
%A valuation $p:P\to\grandqplus$ belong to $b$ denoted $\pval\in b$ if for all $x\in P$, $\pval(x)\in b(x)$.\ea{à virer, j'ai mis une déf plus loin}

\LongVersion{
\subsection{Parametric timed broadcast protocols}
}

We now introduce parametric %\ea{ça devrait être parametric si on suit la litérature : parameterized pour les nombres inconnus de processus (mais ici c'est ``network protocol'' qui inclut cet aspect) et ``parametric'' pour les paramètres temporels (comme dans parametric timed automata et parametric time Petri nets)} 
timed broadcast protocols (PTBP), which are timed broadcast protocols \cite{ADRST11} extended with timing parameters in clock guards.
%which are broadcast protocols studied in \eg{} \cite{DSTZ-fsttcs12,delzanno2011power,delzanno2011parameterized} but extended with clock variables\ea{ben bof, les TBP ont déjà été définis quelque part quand meme ! on ajoute simplement les paramètres} and parametric guards.
Equivalently, PTBP can be seen as a PTA~\cite{AHV93} augmented with communication features.%\ea{rajouté ``an arbitrary number of'' parce que sinon justement c'est pas vrai}

\begin{definition}[Parameterized timed broadcast protocol]
	A \emph{Parameterized timed broadcast protocol} (PTBP) is a tuple $\PTBP=(\Loc,\Clock,\Alphabet,\Param,\locinit,\Edges)$ where:
	\begin{itemize}
	\item $\Loc$ is a finite set of states; 
	\item $\Clock$ is a finite set of clocks; 
	\item $\Alphabet$ is the finite communication alphabet; 
	\item $\Param$ is a finite set of timing parameters; 
	\item $\locinit\in \Loc$ is the initial state; and 
	\item $\Edges\subseteq \Loc\times \GuardsXP \times Act \times \Updates(\Clock) \times \Loc$ is the edge relation, where $Act$ is the set of actions composed of:
	\begin{itemize}
	\item internal actions:  $\epsilon$; 
	\item broadcasts of a message $m\in\Alphabet$: $!!m$; and 
	\item reception of a message $m\in \Alphabet$: $??m$.
	\end{itemize}

	\end{itemize}

%	\begin{itemize}
%		\item $\Loc$ is a finite set of states
%		\item $\Clock$ is a finite set of clocks
%		\item $\Alphabet$ is the finite communication alphabet
%		\item $\Param$ is a finite set of timing parameters
%		\item $\locinit\in \Loc$ is the initial state
%		\item $\Edges\subseteq \Loc\times \GuardsXP \times Act \times \Updates(\Clock) \times \Loc$ is the edge relation, where $Act$ is the set of actions composed of:
%		\begin{itemize}
%			\item internal actions:  $\epsilon$
%			\item broadcasts of a message $m\in\Alphabet$: $!!m$
%			\item reception of a message $m\in \Alphabet$: $??m$
%		\end{itemize}
%		\end{itemize}
\end{definition}
%We denote by PTBP-$n$ the set of PTBP with $n$ clocks.\ea{ouh c'est moche… je préfère explicitement with $n$ clocks, parce que sinon on va confondre les participants et les horloges (mais bon)}
A PTBP is a U-PTBP, L-PTBP, or L/U-PTBP if
	all timing parameters appear only as upper bounds in guards (\ie{} of the form $\clock < \param$ or $\clock \leq \param$),
	only as lower bounds (\ie{} of the form $\clock > \param$ or $\clock  \geq \param$),
	or if the set of parameters~$\Param$ is partitioned into lower-bound and upper-bound parameters,
	respectively.
	
A \emph{bounded} PTBP is a pair $(\PTBP,\bounds)$ where $\PTBP$ is a PTBP and $\bounds: \Param \rightarrow \interval_{\grandqplus}$ are bounds on the parameters that assign to each parameter~$\param$ an interval
	$[\boundinf, \boundsup]$,
		$(\boundinf, \boundsup]$,
		$[\boundinf, \boundsup)$,
		or
		$(\boundinf, \boundsup)$,
	with $\boundinf, \boundsup \in \grandn$.
	We use $\boundinfin{\param}{\bounds}$ and $\boundsupin{\param}{\bounds}$ to denote the infimum and the supremum of~$\param$, respectively.	
A bounded PTBP is a \emph{closed PTBP} if, for each parameter $\param$, its ranging interval $\bounds(\param)$ is of the form $[\boundinf, \boundsup]$. Otherwise it is open bounded.
Abusing notation we say that a parameter valuation~$\pval$ belongs to a bound $\bounds$, written $\pval\in\bounds$, if for all parameters $\param$, $\pval(\param)\in\bounds(\param)$.

\ea{il faudra citer \cite{HRSV02} quelque part (pas forcément là)}

%\ea{copié-collé d'ailleurs: à revoir et étendre aux PTBP (et réduire ?) quand les notations ci-dessus auront stabilisé:
%In this manuscript, we will also consider \emph{bounded} PTAs, \ie{} PTAs with a bounded parameter domain that assigns to each parameter an infimum and a supremum, both integers.
%
%%------------------------------------------------------------
%\begin{definition}[bounded PTA]\label{definition:boundedPTA}
%	A \emph{bounded PTA} is $\bounded{\A}{\bounds}$, where $\A$ is a PTA, and $\bounds : \Param \rightarrow \interval(\grandn)$ assigns to each parameter~$\param$ an interval
%	$[\boundinf, \boundsup]$,
%		$(\boundinf, \boundsup]$,
%		$[\boundinf, \boundsup)$,
%		or
%		$(\boundinf, \boundsup)$,
%	with $\boundinf, \boundsup \in \grandn$.
%	We use $\boundinfin{\param}{\bounds}$ and $\boundsupin{\param}{\bounds}$ to denote the infimum and the supremum of~$\param$, respectively.
%	%
%	(Note that we rule out $\infty$ as a supremum.)
%	
%	We say that a bounded PTA is a \emph{closed bounded PTA} if, for each parameter~$\param$, its ranging interval $\bounds(\param)$ is of the form $[\boundinf, \boundsup]$; otherwise it is an \emph{open bounded PTA}.
%	
%	We define similarly bounded L/U-PTAs.
%\end{definition}
%%------------------------------------------------------------
%
%}

%\ea{c'est quand même vache de pas donner un seul exemple dans cette section !}

\begin{example}
\label{ex:PTBP}
An example of a PTBP is given in \cref{fig:PTBP}. This PTBP is composed of an initial state $q_0$, two states $\factory$ and $\client$ representing a factory and a client, three counting states $1$, $2$ and $3$ and a goal state $\goal$. The set of clocks is the singleton $\{x\}$ and the communication alphabet is composed of two messages $\produ$ and $f$.
There are two timing parameters $pt$ and $tl$ representing respectively the production time and the time limit. Notice that this PTBP is in fact an L/U-PTBP since the parameter $pt$ appears only in guards as a lower bound and $tl$ only as an upper bound.

\begin{figure}[tb]
\centering
\scalebox{.8}{
\begin{tikzpicture}[node distance=2cm]
\node[draw,circle,initial,initial text=] (q0) {$q_0$};
\node[draw,circle] (client) [above right of =q0,yshift=-0.5cm] {\client};
\node[draw,circle] (facto) [below right of =q0,yshift=0.5cm] {\factory};
\node[draw,circle] (1) [right of = client] {1};
\node[draw,circle] (2) [right of = 1] {2};
\node[draw,circle] (3) [right of = 2] {3};
\node[draw,circle] (goal) [right of = 3] {\goal};

\path[draw,->,above]
(q0) edge[sloped] node {$??f$} (client)
(q0) edge[sloped] node {$!!f,\{x\}$} (facto)
(client) edge node {$??\produ$} (1)
(1) edge node {$??\produ$} (2)
(2) edge node {$??\produ$} (3)
(3) edge node {$x< tl,\epsilon$} (goal)

(facto) edge[loop right] node[right] {$x\geq pt,!!\produ,\{x\}$} (facto)
;
\end{tikzpicture}
}

\caption{Example of a (L/U-)PTBP\label{fig:PTBP}}
\end{figure}
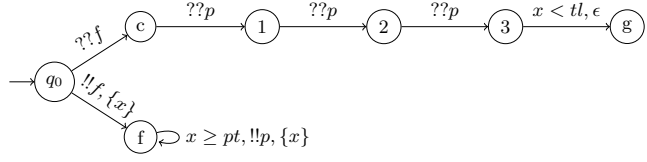
\end{example}

%A network of size~1 is a PTBP-$n$ with a single process.\ea{il faudrait définir une notation, pour pouvoir être utilisée dans la déf des problèmes ci-dessous}
%\ea{network n'est pas défini}

\LongVersion{
\subsection{Networks}\label{sec:sem}
}
\ea{il faut donner l'équivalence d'un PTA à partir d'un PTBP quand le nombre de participants (concept indéfini pour l'instant) est~1.
Il faut aussi donner la déf (au moins son nom) des L/U-, U- et L-PTA puisqu'on les utilise}%
\pf{voir la remark a la fin de 3.1}%
We now define the semantics of parameterized networks of PTBP. This semantics is illustrated in \cref{ex:run} after the formal definition.

A network is composed of a multitude of processes all running the same protocol~$\PTBP$.
Let $N$ denote the number of processes, or \emph{size} of the network.

Formally, a configuration $\config$ of a network running a parametric timed broadcast protocol $\PTBP=(\Loc,\Clock,\Alphabet,\Param,\locinit,\Edges)$ is a vector $\config\in\vect(\Loc\times\ClockVals)$.
%\ea{never used: ``denote by $\Configs$ the set of all such configurations.''}
Intuitively, a configuration $\config$ with $\config[\pid]=(\loc,\clockval)$ means that the process $\pid$ is in state~$\loc$ and with clock valuation~$\clockval$.

Given a configuration $\config$ with $N$ processes and a process $\pid$, we write $\state(\config[\pid])$ for the state and $\val(\config[\pid])$ for the valuation such that $\config[\pid]=(\state(\config[\pid]),\val(\config[\pid]))$.
Abusing notation we extend $\state$ to the whole configuration \ie{}
% 	$\state(\config)=\cup_{\pid\in\{1,..,N\}}\state(\config[\pid])$.\bd{C'est pas un vecteur ? Vraiment un ensemble ?}
%% CHANGED BY ÉA July 2018
	$\state(\config)[i] = \state(\config[\pid])$ for $i \in \{ 1, \dots, N \}$.\bd{C'est pas un vecteur ? Vraiment un ensemble ?}\ea{fixed?}

Note that the representation of configuration as vectors is only for practical reasons, the processes are identical and do not have ids.

We say that a configuration $\config$ is initial if all processes are in  the initial state and their clocks are all set to $0$ \ie{} for all $\pid$, $\config[\pid]=(\locinit,\clockvalzero)$.

Given a timing parameter valuation\LongVersion{ $\pval$}, the transition relation on configurations is intuitively defined as follows:
First a delay is chosen and all the clocks in the network are increased by this delay.
Then one of the processes performs a possible action \ie{} an action for which the guard is satisfied given its clock valuation and the valuation of the timing parameter. Two cases follow. Either the action is internal and only this process moves and updates its clocks accordingly, or the action is a broadcast and a set of receivers is chosen.
It this latter case, the sender moves and updates its clocks and all the chosen receivers also move and update their clocks accordingly.

More formally, given a timing parameter valuation $\pval$ and a configuration $\config\in \vect_N(\Loc\times\ClockVals)$, there are transitions for all $t\in\grandrplus$, $\pid\in\{1,\dots,N\}$, $\edge=(\loc_1, \guard, \action, \update, \loc_2)\in\Edges$, and $R\subseteq \{1,..,N\}$ such that:
\begin{description}
\item[elapse of time] there is a valuation $\config_t\in \vect_N(\Loc\times\ClockVals)$ such that $\forall j\in\{1,\dots,N\}$, $\config_t[j]=(\loc,\clockval+t)$ where $(\loc,\clockval)=\config[j]$, and
\item[execution of the action] the following conditions are satisfied:
\begin{description}
 \item[the action is enabled] $\state(\config_t[\pid])=\loc_1$ and $\val(\config_t[\pid])\models_{\pval} \guard$, and
 \item[execution of the action] the transition leads to a configuration $\config'$ such that
 	\begin{itemize}
 	\item the active process performed the action: $\config'[\pid]=(\loc_2,\update(\val(\config_t[\pid])))$,
 	\item unconcerned processes are unaffected: $\forall j\in\{1,\dots,N\}\setminus(R\cup\{\pid\})$,  $\config'[j]=\config_t[j]$, and
 	\item either 
 	\begin{itemize}
 		\item $a$ is an internal action ($a=\epsilon$) and the receiving processes are unaffected: $\forall j\in R\setminus\{\pid\}$,  $\config'[j]=\config_t[j]$; or 
 		\item $a=!!m$ and $\forall j\in R\setminus\{\pid\}$, if there exists an edge $(\state(\config_t[j]),\guard',\linebreak[3]??m,\update',q')$ such that $\val(\config_t[j])\models_{\pval} \guard'$, then the process receives the message and $\config'[j]=(q',\update'(\val(\config_t[j])))$.  Otherwise the process is unaffected and $\config'[j]=\config_t[j]$.
 		\end{itemize}
\end{itemize} 	 
 \end{description}
\end{description}
When such a transition exists, it is written $\config\xrightarrow{t,\pid,\edge,R}_{\pval} \config'$ or simply $\config\to_{\pval} \config'$.

Notice that we consider non blocking broadcast \ie{} if a process is in the receiver set but has no available reception edge, the process is unaffected and the network behaves as if this process was not in the receiver set.

An execution $\rho$ is a sequence of transitions starting in an initial configuration $\config_0$, $\rho=\config_0\to_{\pval} \config_1\to_{\pval} \cdots$.
An execution is maximal if it is infinite or if it ends in a configuration from which there is no possible transition.

Notice that once an initial configuration is fixed, the number of processes does not change along an execution. However the semantics is infinite for several reasons: first there is an infinite number of initial configurations (\ie{} of network sizes); second, there is also an infinite number of possible parameter valuations; third, given a network size and parameter valuation, clock valuations assign real values to clocks and are thus uncountable.

Given  PTBP $\PTBP$, a network size $N$ and a timing parameter valuation $\pval$, we denote by $\Execs(\PTBP,N, \pval)$ the set of all maximal executions for the valuation~$\pval$ with $N$~processes.

We say that a maximal execution $\rho=\config_0\to_{\pval}\config_1\to_{\pval} \cdots$ reaches a state~$\loc$, written $\rho\models \Diamond \loc$, if there exists an index $n$ such that $\loc \in \state(\config_{n})$. 

\begin{example}\label{ex:run}
We give an example of  a possible execution for a network composed of $4$~processes running the protocol given in \cref{ex:PTBP}. In this example $tl=9$ and $pt=3$. The edge used during a transition is here only represented by the associated action for readability.
\begin{scriptsize}
\begin{align*}
\scriptscriptstyle
\left( \begin{array}{c}
q_0,0 \\
q_0,0 \\
q_0,0 \\
q_0,0 \\
q_0,0 
\end{array} \right)
\xrightarrow{0.1,1,f,\emptyset} 
\left( \begin{array}{c}
\factory,0 \\
q_0,0.1 \\
q_0,0.1 \\
q_0,0.1 \\
q_0,0.1 
\end{array} \right)
\xrightarrow{4.1,2,f,\{3,5\}} 
\left( \begin{array}{c}
\factory,4.1 \\
\factory,0 \\
\client,4.2 \\
q_0,4.2 \\
\client,4.2 
\end{array} \right)
\xrightarrow{1.3,1,\produ,\{5\}} 
\left( \begin{array}{c}
\factory,0 \\
\factory,1.3 \\
\client,5.5 \\
q_0,5.5 \\
1,5.5 
\end{array} \right)
\xrightarrow{1.8,2,\produ,\{1,3,4,5\}}
\\
\xrightarrow{1.8,2,\produ,\{1,3,4,5\}}
\left( \begin{array}{c}
\factory,1.8 \\
\factory,0 \\
1,7.3 \\
q_0,7.3 \\
2,7.3 
\end{array} \right)
\xrightarrow{1.2,1,\produ,\{5\}} 
\left( \begin{array}{c}
\factory,0 \\
\factory,1.2 \\
1,8.5 \\
q_0,8.5 \\
3,8.5 
\end{array} \right)
\xrightarrow{0,5,\epsilon,\emptyset} 
\left( \begin{array}{c}
\factory,0 \\
\factory,1.2 \\
1,8.5 \\
q_0,8.5 \\
\goal,8.5 
\end{array} \right)
\end{align*}

\end{scriptsize}

\end{example}

\begin{remark}
Notice that even if the notations are slightly different, PTBP networks fully extend both PTA~\cite{AHV93} and timed broadcast protocols~\cite{ADRST11}. Indeed,  PTA are PTBP networks of size one and timed networks are PTBP networks without timing parameters.
\end{remark}

\LongVersion{
\paragraph*{Problems considered}
}
In this paper, we consider parameterized  reachability
	%\ea{oui mais est-ce que AF est vraiment de la ``reachability''…?}\pf{Ben oui, non? reachability pour tout les run}
problems: we ask whether there exists a network size $N$ satisfying a given reachability property.
We consider existential (EF) and universal (AF) reachability properties that ask, given goal state $\locfinal$, whether this state is reached by some (EF) or all (AF) executions.
Moreover we also consider variants on the quantifier on timing parameters and ask whether the property holds for all parameter valuations (universality) or for at least one (existence).

Thus, given a bounded PTBP $(\PTBP,\bounds)$ and a state $\locfinal$ %and parameter bound $b$\ea{on n'a pas toujours de borne ; en plus, $\pval \in b$ ce n'est pas défini}  
we consider the following problems:
\begin{description}
	\item[$\exists$-EF-existence]~~$\exists N\in\grandn, \exists \pval\in \bounds,\exists\rho\in\Execs(\PTBP,N,\pval), \rho\models \Diamond \locfinal$
	\item[$\exists$-EF-universality] $\exists N\in\grandn,\forall \pval\in \bounds,\exists\rho\in\Execs(\PTBP,N,\pval), \rho\models \Diamond \locfinal$
	\item[$\exists$-AF-existence]~$\exists N\in\grandn, \exists \pval\in \bounds,\forall\rho\in\Execs(\PTBP,N,\pval), \rho\models \Diamond \locfinal$
	\item[$\exists$-AF-universality] $\exists N\in\grandn,\forall \pval\in \bounds,\forall\rho\in\Execs(\PTBP,N,\pval), \rho\models \Diamond \locfinal$
\end{description}
Note that, in contrast to PTAs, where the emptiness problem (the emptiness of the valuation set for which a property holds) is equivalent to the existence problem; this is not the case in our setting, because of the additional quantifier on the network size (``$\exists N\in\grandn$'').

For convenience, we will omit the bounds when they are irrelevant and consider these problems in the case of general PTBP. In the following, the bounds will only be relevant in \cref{section:LU}.

In the next section we investigate these problems in the general semantics defined above.
This semantics is called \emph{reconfigurable} since the communication topology (modeled by the reception sets) can be reconfigured at each step.
However, in broadcast protocol networks with a parametric number of processes, the communication topology plays a decisive role on decidability status.
We will thus investigate another communication setting, in \cref{section:clique}, in which every message is received by all the other processes \ie{} the reception set $R$ is always equal to $\{1,\dots,N\}$.
These networks are called \emph{clique} networks.

\begin{example}
Considering the PTBP given in \cref{ex:PTBP} and the target state $\goal$.
The execution presented in \cref{ex:run} shows that the answer for the $\exists$-EF-existence problem is positive whenever the bounds allow for $tl=9$ and $pt=3$ in the reconfigurable semantics.
Notice that in the clique semantics, it is not possible to reach $\goal$ unless $pt*3<tl$. Indeed in the clique semantics when a first process moves to $\factory$, all the other processes receive the message $f$ and thus move to $\client$. Thus, at least three $pt$ time units are necessary in order to receive 3 messages $\produ$.

Notice also that in this example, in both semantics, both $\exists$-AF problems would give negative answers since there is always an execution that forever sends $p$ in the bottom self-loop and never uses the internal transition leading to $\goal$. Thus such an execution never reaches $\goal$.
\end{example}

%%%%%%%%%%%%%%%%%%%%%%%%%%%%%%%%%%%%%%%%%%%%%%%%%%%%%%%%%%%%
%%%%%%%%%%%%%%%%%%%%%%%%%%%%%%%%%%%%%%%%%%%%%%%%%%%%%%%%%%%%
\section{Reconfigurable semantics}\label{section:reconfigurable}
%%%%%%%%%%%%%%%%%%%%%%%%%%%%%%%%%%%%%%%%%%%%%%%%%%%%%%%%%%%%
%%%%%%%%%%%%%%%%%%%%%%%%%%%%%%%%%%%%%%%%%%%%%%%%%%%%%%%%%%%%

\ea{c'est bizarre de commencer par AF puis de faire EF (même si je comprends qu'il vaut mieux garder le meilleur pour la fin)}

\ea{je préciserais dans tous les intitulés ``for/in the reconfigurable semantics'', puisqu'on le fait aussi en \cref{section:clique}}

%%%%%%%%%%%%%%%%%%%%%%%%%%%%%%%%%%%%%%%%%%%%%%%%%%%%%%%%%%%%
\subsection{AF problems in the reconfigurable semantics}
%%%%%%%%%%%%%%%%%%%%%%%%%%%%%%%%%%%%%%%%%%%%%%%%%%%%%%%%%%%%

The reconfigurable semantics of broadcast networks, where the set of receivers can be chosen non-deterministically, makes the AF problems equivalent to the same problems in networks of size~1.
This is due to the fact that in the reconfigurable semantics nothing prevents messages to be sent to an empty set of receivers.
The following theorem is a direct consequence of previous known results on parameterized timed automata and this previous remark\bd{ (the detailed proof can be found in \cref{proof:thm:reconf-AF})}.\ea{puisque c'est direct (et je suis d'accord que ça l'est), peut-être que de la mettre en annexe serait une bonne idée, si on manque de place}\footnote{%
	The proof of the results that can be obtained using existing techniques in a more or less straightforward manner can be found in the appendix.
}

\ea{le théorème suivant est vrai en borné ou pas pour 1 horloge ; pour 3 horloges, c'est moins sûr, il faut revérifier la littérature (on a peut-être besoin de plus que 3 horloges ?)}

%\pf{Je suis pas sur des ref (ni des resultats)… Vous pouvez checker?}\ea{j'ai revu l'intitulé}
\newcommand{\thmReconfAF}{$\exists$-AF-existence and $\exists$-AF-universality are \dec{decidable} for 1 clock PTBP but \undec{undecidable} for (L/U)-PTBP with 3 clocks or more.}
%------------------------------------------------------------
\begin{theorem}\label{thm:reconf-AF}
	\thmReconfAF{}
\end{theorem}
%------------------------------------------------------------

%%%%%%%%%%%%%%%%%%%%%%%%%%%%%%%%%%%%%%%%%%%%%%%%%%%%%%%%%%%%
\subsection{EF problems in the reconfigurable semantics}
%%%%%%%%%%%%%%%%%%%%%%%%%%%%%%%%%%%%%%%%%%%%%%%%%%%%%%%%%%%%

We start by recalling some known results on networks composed of an arbitrary number of timed processes.
In~\cite{AJ03} the authors considered timed networks and proved that the reachability problem ($\exists$-EF) is decidable \LongVersion{when considering network of processes }with one clock per process and undecidable for two clocks per process~\cite{ADM04}.
Note that timed networks have a different semantics than the one we use in this paper since they use rules and not broadcasts. However the reconfigurable semantics can be easily encoded in the rules of timed networks.
%Even if time networks have a slightly different semantics than ours network, in the case without timing parameter and with one clock per process, one could encode our semantics in the rules they use for theirs.\ea{mouaif mouaif… je serais reviewer, je trébucherais un peu sur cette phrase.}
This gives us the decidability of the $\exists$-EF problem (without timing parameters and with one clock per process).

%------------------------------------------------------------
\begin{theorem}[\cite{AJ03,ADM04}]\label{thm:reconf-EFTA}
	$\exists$-EF is \dec{decidable} for PTBP without parameters and with one clock per process and \undec{undecidable} with two clocks per process.
\end{theorem}
%------------------------------------------------------------

A direct consequence of this theorem is the undecidability of the $\exists$-EF problems for PTBP with two clocks.

%------------------------------------------------------------
\begin{lemma}\label{lemma:EF:undec:2}
	The $\exists$-EF-existence and $\exists$-EF-universality problems are \undec{undecidable} for PTBP with two clocks.
\end{lemma}
%------------------------------------------------------------

Moreover, we show below that the undecidability even holds for PTBP with a single clock.
This is a major difference with both parameterized networks and PTA, where the restriction to one clock leads to decidability~\cite{AHV93}.
Also observe that our result does not rely on the reconfigurable semantics particularly.

\ea{notes pour moi-même :
\begin{enumerate}
	\item je n'ai pas relu la preuve ;
	\item ça contredit (sauf erreur) notre discussion de mai 2017 : `` PTA à 1 seule horloge par composant : décidable. Paulin dit que, si la sémantique d'un composant est finie, alors on peut raisonner sur le graphe des zones et, modulo une combinatoire (actions ? autorisées), on termine. Or le graphe des zones des zones paramétrées d'un PTA à 1 horloge est fini ; on aurait donc non seulement le vide (y a-t-il au moins une valeur des paramètres…?) mais peut-être aussi la synthèse (pour quelles valeurs…?) Trivialement transposable aux L/U, L- et U- à 1 horloge.''
\end{enumerate}

}
\pf{oui j'avai dit des bétises. meme si le graph des regiion est fini pour un composant, il n'est pas du tout fini pour un nombre arbitraire de composant vu qu'il y a un nobre arbitraire d'horloge et qu'il faut considerer leur ordre. Je pensai que c'étai pas trop grave mais en fait si.}

\ea{le théorème suivant est uniquement vrai en \emph{non-borné} ; on pourrait imaginer l'adapter en borné en reprenant la preuve de \cite{ALR16ICFEM} mais je n'ai aucune idée si ça marche effectivement (et je crois que c'est de la borne ouverte de mémoire ? même si on doit pouvoir adapter en fermé avec un gadget initial, à voir en distribué ???)}

\newcommand{\thmReconfPTBP}{The $\exists$-EF-existence and $\exists$-EF-universality problems are \undec{undecidable} for PTBP with one clock.}
%------------------------------------------------------------
\begin{theorem}\label{thm:reconf-PTBP}
	\thmReconfPTBP{}
\end{theorem}

\begin{proof}
%Recall that we show here the undecidability of the $\exists$-EF-emptiness and $\exists$-EF-universality for PTBP with one clock in the reconfigurable semantics.
The proof is by reduction of the halting and boundedness (respectively) problems for two-counter machines\LongVersion{ (recalled in \cref{appendix:2CM})}.

First, in this proof we will assume that the parameter $\param$ only takes integer values. This is not a restriction since we can add a gadget at the beginning of the PTBP to check such property. This gadget is an adaptation of similar gadgets from the PTA community to the case of PTBP, and is given in \LongVersion{\cref{fig:gadget} in \cref{appendix:proof:thm:reconf-PTBP-details:integer}}\ShortVersion{\cite{ADFL18}}.

Given a two-counter machine, we define a protocol $\mathcal{P}$ separated in three parts, the controller part (in charge of tracking the current instruction), the counters part (to model the counters behaviors) and an idle part that allows to use additional processes when needed.

The value of the counters is encoded (up to the value of parameter $\param$ minus~$1$ here for technical reasons) by the difference between the clock value of the processes in states representing counters and the clock value of the processes in the controller part.

Formally, $\mathcal{P}$ is defined as follows:
\begin{itemize}
	\item $\Loc=\{\locinit,error,c_i,nc1_i,nc2_i,zt1_i^j,zt2_i^j,dec1_i^j,dec2_i^j,inc1_i^j,inc2_i^j,inc3_i^j,idle \mid j\in\{1,2\},i\in\{1,2\}\}\cup \{k^j\mid \TwoCMinstruction \in \TwoCMInstructions,j\in\{1,2,3,4\}\}$
	\item $\Alphabet = \{tick,inc_i,dec_i,zt_i,c_i,oc_i,nc_i\mid i\in \{1,2\}\}$
	\item $\Param=\{\param \}$
	\item $\Clock=\{\clock \}$
	\item $\Edges$ is defined as described below.
\end{itemize}
Let us describe $\Edges$:
	On every transition, there is a guard $\clock \leq \param$ which is omitted to clarify notations; similarly, when a guard is true\LongVersion{ (here limited to $\clock \leq \param$)} or when there is no reset, we omit them in the transition.
The construction is represented in \cref{fig:const}. $\Edges$ is composed of the following transitions:
\begin{description}
	%------------------------------------------------------------
	\item[Initialization.] $(\locinit,\clock=0,\epsilon,k_0^1)$, for $i\in\{1,2\}$, $(\locinit,\clock=0,\epsilon,c_i)$, $(\locinit,\clock=0,\epsilon,idle)$\\
	The processes can chose non-deterministically to either move to the controller part, the counters part, or the idle part (\cref{figure:construction:initialization}).
	\item[Decrement of counter $i$.] For a decrement instruction $\TwoCMinstruction:decr~\TwoCMCounter_i~goto~\TwoCMinstruction_1$, we define the following transitions in $\Edges$ (depicted in \cref{figure:construction:decrement}):
		\begin{itemize}
		\item For the controller: $(k^1,\clock=1,!!dec_i,k^2)$ $(k^2,\clock=\param,!!tick,\{\clock:=0\},k_1^1)$\\
		The controller ``announces'' that the instruction is a decrement (using $!!dec_1$) when its clock is equal to $1$ (guard $\clock = 1$) and then announces when its clock reaches the value of the parameter (guard $\clock = \param$).
		\item For the counter involved ($i$): $(c_i,\clock> 1,??dec_i,dec1_i^i)$,  $(dec1_i^i,\clock=\param,\epsilon,\{\clock:=0\},dec2_i^i)$ $(dec2_i^i,\clock=1,\epsilon,\{\clock:=0\},dec3_i^i)$ $(dec3_i^i,??tick,c_i)$\\
		When the processes representing the counter~$i$ receive the message corresponding to the decrement, they move to an intermediary state, then reset their clock when it reaches $\param$ and reset it another time when the clock reaches~$1$.
		This way, the difference with the controller clock has decreased by one.
		Notice that, if $x=1$ when they receive the decrement message (meaning that the counter has value~0), they cannot take the transition.
		\item For the counter not involved ($3-i$): $(c_j,??dec_i,dec1_i^j)$ $(dec1_i^j,\clock=\param,\{\clock:=0\},decj2_i^j)$ $(dec2_i^j,??tick,c_j)$.\\
		The processes encoding the counter not involved just reset their clock when it reaches~$\param$, thus the difference remains constant.
	\end{itemize}
	%------------------------------------------------------------
	\item[Increment of counter $i$.] for an increment instruction $\TwoCMinstruction:incr~\TwoCMCounter_i~goto~\TwoCMinstruction_1$,
		the construction is almost symmetric to decrement, but involves an additional technicality---and therefore we give it below.
	We define the following transitions in $\Edges$ (depicted in \cref{figure:construction:increment}):
	\begin{itemize}
		\item For the controller: $(k^1,\clock=1,!!inc_i,k^2)$ $(k^2,\clock=\param,!!tick,\{\clock:=0\},k_1^1)$\\
		The controller announces that the instruction is an increment when its clock is equal to $1$ and then announces when its clock reaches the value of the parameter.
		\item For the counter involved: \\
		The clock value should be reset at $\param-1$, but such a guard is not allowed and is not possible to encode with just one clock.
			As an additional technicality, we thus rely on a \emph{non-deterministic guess}, that is the checked by a new process.
			This is done as follows:
		\begin{description}
			\item[For the current counter processes] $(c_i, \clock<\param,??inc_i,inc1_i^i)$,$(c_i, \clock=\param,??inc_i,error)$, $(inc1_i^i,!!nc_i,inc2_i^i)$ $(inc2_i^i,\clock=\param,!!oc_i,idle)$.\\
			The processes encoding the counter receive the increment message and then guess non-deterministically that their clock value is $\param-1$ and send a message $nc_i$.
			In order to check that the guess was right, they then announce when their clock reaches $\param$ by sending message $oc_i$, and the processes move to~$idle$.
			The value of the counter will then be encoded by the new processes.
			Notice that if the clock value is already equal to~$\param$, then we reached the maximal possible value, and the processes move to the error state~$error$.
			\item[For the new counter process] $(idle,??nc_i,\{\clock:=0\},nc1_i)$ $(nc_i,\clock=1,??oc_i,nc2_i)$ $(nc2_i,??tick,c_i)$.\\
			To check that the guess was right, we use the idle processes that when receiving the message $nc_i$ reset their clock.
				They are then allowed to encode the counter if they receive the confirmation $oc_i$ when their clock is equal to~$1$ (thus the guess was correct).

		\end{description}
		\item For the counter not involved: $(c_j,??inc_i,inc1_i^j)$ $(inc1_i^j,\clock=\param,\{\clock:=0\},incj2_i^j)$ $(inc2_i^j,??tick,c_j)$.\\
			The processes encoding the counter not involved just reset their clock when it reaches~$\param$.
	\end{itemize}
	%------------------------------------------------------------
	\item[Zero-test.] For a zero-test instruction $\TwoCMinstruction:if~\TwoCMCounter_i=0~then~goto~\TwoCMinstruction_1~else~goto~\TwoCMinstruction_2$, we define the following transitions:
	\begin{itemize}
		\item For the controller $(k^1,\clock=1,!!zt_i,k^2)$ $(k^2,\clock=\param,??c_i,k^3)$ $(k^2,\clock<\param,\linebreak[3]??c_i,k^4)$ $(k^3,\clock=\param,!!tick,\{\clock:=0\},k_1^1)$ $(k^4,\clock=\param,!!tick,\{\clock:=0\},k_2^1)$.\\
		The controller announces that the instruction is a zero-test when its clock is equal to~$1$, and then waits for a notification $c_i$ from the counter.
		Depending when this notification arrives, when $\clock=\param$ (meaning the counter has value $0$) or when $\clock<\param$ (meaning the counter has positive value), the controller moves to the corresponding intermediary states.
		\item For the counter involved $(c_i,??zt_i,zt1_i^i)$ $(zt1_i,\clock=\param,!!c_i,\{\clock:=0\},zt2_i^i)$ $(zt2_i^i,??tick,c_i)$.\\
		The processes encoding the counter involved, after receiving the instruction, send a notification $c_i$ when their clock reaches~$\param$.
		\item For the counter not involved $(c_j,??zt_i,zt1_i^j)$ $(zt1_i,\clock=\param,\epsilon,\{\clock:=0\},zt2_i^j)$ $(zt2_i^j,??tick,c_j)$.\\
		The processes encoding the counter not involved just reset their clock when it reaches~$\param$.
	\end{itemize}
	%------------------------------------------------------------
	\item[] Finally, there is an additional transition $(idle,\epsilon,\{\clock:=0\},idle)$ used to keep the clock of idle processes below~$\pval(\param)$.
\end{description}

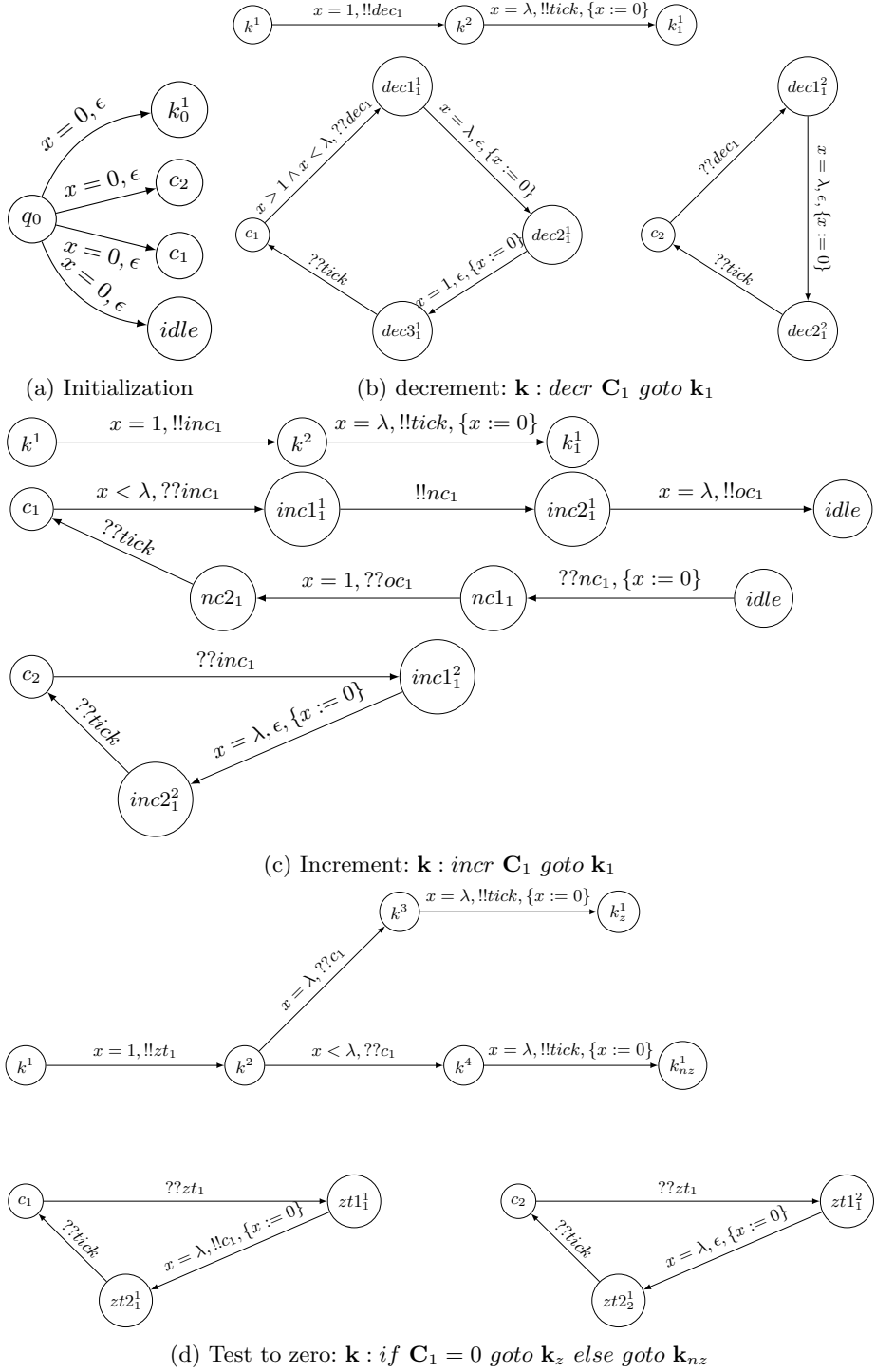
\begin{figure}
	\begin{subfigure}[b]{0.25\textwidth}
	\resizebox{\textwidth}{!}{
	\begin{tikzpicture}

		\node[draw,circle] (0) {$\locinit$};
		\node[draw,circle] (c1) [right of=0,xshift=1cm,yshift=-0.5cm] {$c_1$};
		\node[draw,circle] (c2) [above of=c1] {$c_2$};
		\node[draw,circle] (k) [above of=c2] {$k_0^1$};
		\node[draw,circle] (i) [below of=c1] {$idle$};

		\path[-latex,sloped,above] (0)
			edge node[below] {$\clock=0,\epsilon$} (c1)
			edge node {$\clock=0,\epsilon$} (c2)
			edge[bend left] node {$\clock=0,\epsilon$} (k)
			edge[bend right] node {$\clock=0,\epsilon$} (i);
	\end{tikzpicture}
	}
	\caption{Initialization}
	\label{figure:construction:initialization}
	\end{subfigure}
	%------------------------------------------------------------
	\begin{subfigure}[b]{0.7\textwidth}
	\resizebox{\textwidth}{!}{
	\begin{tikzpicture}[node distance=4cm]
		\node[draw,circle] (k1) {$k^1$};
		\node[draw,circle] (k2) [right of=k1] {$k^2$};
		\node[draw,circle] (k3) [right of=k2] {$k_1^1$};
		\path[-latex,sloped,above]
		(k1) edge node {$\clock=1,!!dec_1$} (k2)
		(k2) edge node {$\clock=\param,!!tick,\{\clock:=0\}$} (k3);

		\node[draw,circle] (c1) [below of=k1] {$c_1$};
		\node[draw,circle] (i1) [above right of=c1] {$dec1_1^1$};
		\node[draw,circle] (i2) [below right of=i1] {$dec2_1^1$};
		\node[draw,circle] (i3) [below left of=i2,yshift=1cm] {$dec3_1^1$};
		\path[-latex,sloped,above]
		(c1) edge node {$\clock>1\wedge\clock<\param,??dec_1$} (i1)
		(i1) edge node {$\clock=\param,\epsilon,\{\clock:=0\}$} (i2)
		(i2) edge node {$\clock=1,\epsilon,\{\clock:=0\}$} (i3)
		(i3) edge node {$??tick$} (c1);

		\node[draw,circle] (c2) [right of=i2,xshift=-2cm] {$c_2$};
		\node[draw,circle] (i12) [above right of=c2] {$dec1_1^2$};
		\node[draw,circle] (i22) [below right of=c2,yshift=1cm] {$dec2_1^2$};
		\path[-latex,sloped,above]
		(c2) edge node {$??dec_1$} (i12)
		(i12) edge node {$\clock=\param,\epsilon,\{\clock:=0\}$} (i22)
		(i22) edge node {$??tick$} (c2);

	\end{tikzpicture}
	}
	\caption{decrement: $\TwoCMinstruction:decr~\TwoCMCounter_1~goto~\TwoCMinstruction_1$}
	\label{figure:construction:decrement}
	\end{subfigure}

	%------------------------------------------------------------
	\begin{subfigure}[b]{\textwidth}
	\resizebox{\textwidth}{!}{
	\begin{tikzpicture}[node distance=4cm]
		\node[draw,circle] (k1) {$k^1$};
		\node[draw,circle] (k2) [right of=k1] {$k^2$};
		\node[draw,circle] (k3) [right of=k2] {$k_1^1$};
		\path[-latex,sloped,above]
		(k1) edge node {$\clock=1,!!inc_1$} (k2)
		(k2) edge node {$\clock=\param,!!tick,\{\clock:=0\}$} (k3);

		\node[draw,circle] (c1) [below of=k1,yshift=3cm] {$c_1$};
		\node[draw,circle] (i1) [right of=c1] {$inc1_1^1$};
		\node[draw,circle] (i2) [right of=i1] {$inc2_1^1$};
		\node[draw,circle] (i3) [right of=i2] {$idle$};
		\path[-latex,sloped,above]
		(c1) edge node {$\clock<\param,??inc_1$} (i1)
		(i1) edge node {$!!nc_1$} (i2)
		(i2) edge node {$\clock=\param,!!oc_1$} (i3);

		\node[draw,circle] (nc2) [below right of=c1,yshift=1.5cm] {$nc2_1$};
		\node[draw,circle] (nc1) [right of=nc2] {$nc1_1$};
		\node[draw,circle] (id) [right of=nc1] {$idle$};
		\path[-latex,sloped,above]
		(id) edge node {$??nc_1,\{\clock:=0\}$} (nc1)
		(nc1) edge node {$\clock=1,??oc_1$} (nc2)
		(nc2) edge node {$??tick$} (c1);

		\node[draw,circle] (c2) [below of=c1,yshift=1.5cm] {$c_2$};
		\node[draw,circle] (i12) [right of=c2,xshift=2cm] {$inc1_1^2$};
		\node[draw,circle] (i22) [below right of=c2,xshift=-1cm,yshift=1cm] {$inc2_1^2$};
		\path[-latex,sloped,above]
		(c2) edge node {$??inc_1$} (i12)
		(i12) edge node {$\clock=\param,\epsilon,\{\clock:=0\}$} (i22)
		(i22) edge node {$??tick$} (c2);

	\end{tikzpicture}
	}
	\caption{Increment: $\TwoCMinstruction:incr~\TwoCMCounter_1~goto~\TwoCMinstruction_1$}
	\label{figure:construction:increment}
	\end{subfigure}

	%------------------------------------------------------------
	\begin{subfigure}{\textwidth}
	\resizebox{\textwidth}{!}{
	\begin{tikzpicture}[node distance=4cm]
		\node[draw,circle] (k1) {$k^1$};
		\node[draw,circle] (k2) [right of=k1] {$k^2$};
		\node[draw,circle] (k3) [above right of=k2] {$k^3$};
		\node[draw,circle] (k4) [right of=k2] {$k^4$};
		\node[draw,circle] (kz) [right of=k3] {$k_z^1$};
		\node[draw,circle] (knz) [right of=k4] {$k_{nz}^1$};

		\path[-latex,sloped,above]
		(k1) edge node {$\clock=1,!!zt_1$} (k2)
		(k2) edge node {$\clock=\param,??c_1$} (k3)
		(k2) edge node {$\clock<\param,??c_1$} (k4)
		(k3) edge node {$\clock=\param,!!tick,\{\clock:=0\}$} (kz)
		(k4) edge node {$\clock=\param,!!tick,\{\clock:=0\}$} (knz)

		;

		\node[draw,circle] (c1) [below of=k1,yshift=1.5cm] {$c_1$};
		\node[draw,circle] (i12) [right of=c1,xshift=2cm] {$zt1_1^1$};
		\node[draw,circle] (i22) [below right of=c1,xshift=-1cm,yshift=1cm] {$zt2^1_1$};
		\path[-latex,sloped,above]
		(c1) edge node {$??zt_1$} (i12)
		(i12) edge node {$\clock=\param,!!c_1,\{\clock:=0\}$} (i22)
		(i22) edge node {$??tick$} (c1);

		\node[draw,circle] (c2) [right of=i12,xshift=-1cm] {$c_2$};
		\node[draw,circle] (i12a) [right of=c2,xshift=2cm] {$zt1_1^2$};
		\node[draw,circle] (i22a) [below right of=c2,xshift=-1cm,yshift=1cm] {$zt2^1_2$};
		\path[-latex,sloped,above]
		(c2) edge node {$??zt_1$} (i12a)
		(i12a) edge node {$\clock=\param,\epsilon,\{\clock:=0\}$} (i22a)
		(i22a) edge node {$??tick$} (c2);

	\end{tikzpicture}
	}
	\caption{Test to zero: $\TwoCMinstruction:if~\TwoCMCounter_1=0~goto~\TwoCMinstruction_z~else~goto~\TwoCMinstruction_{nz}$}
	\end{subfigure}

	\caption{Representation of the construction\label{fig:const}}
\end{figure}

Given a valuation $\pval$ of the parameter, we say that a configuration $\config$ of the network encodes a configuration $(k,v_1,v_2)$ of the two-counter machine if for all~$\pid$, $\config[\pid]=(\loc,\clock)$ then either $\clock>\pval(\param)$ or $\loc \in\{c_1,c_2,k^1,idle\}$.
Moreover all processes with a clock lower than $\pval(\param)$ and not in state $idle$ must agree on their clock valuation if they have the same state.
Finally, if $\config[\pid]=(k^1,z)$ then for all $\pid'$ such that $\config[\pid']=(c_1,y)$ we have $v_1=y-z$ and similarly for~$v_2$.

Given an execution $\rho$, and a time~$t$ we denote by $\rho_{T=t}$ the configuration obtained when considering $\rho$ at global time~$t$.
Notice that $\rho_{T=t}$ may not be a configuration that appears in $\rho$ since it can be a configuration obtain during the elapsing of time in a transition.

We will prove that, for any execution $\rho$, either $\rho_{T=k*\pval(\param)+1/2}$ is not defined (the execution time never reaches $k*\pval(\param)+1/2$) or $\rho_{T=k*\pval(\param)+1/2}$ encodes~$\TwoCMconfig_k$, \ie{} the $k$th configuration of the two-counter machine.
\medskip

We start by some remarks on the shape of possible executions.
\begin{enumerate}
	\item If two processes are in the controller part, then their clocks are equal modulo~$\pval(\param)$.
		Indeed, in the controller part, the clock is reset only when it reaches~$\pval(\param)$.
	\item It follows that, by definition of the protocol, the message $tick$ is sent only at time units multiple of $\pval(\param)$.
	\item Moreover, the instruction messages ($inc_i,dec_i,zt_i$) are only sent at global time units of the form $k*\pval(\param)+1$
	\item Consider a process in state $c_i$ with clock value lower than $\pval(\param)$. Assume that the global time is of the form $k*\pval(\param) +1$. If this process does not receive an instruction message without delay, it will not be able to receive any before time $(k+1)*\pval(\param)+1$, thus it cannot take any transition before $(k+1)*\pval(\param)+1$.
		Note that at this time, its clock will be greater than $\pval(\param)$, thus the guard prevents it to take any transition for the rest of the execution.
	\item With the same idea, if the process is in an intermediary state $nc2_i,dec2_i^j,\linebreak[3]dec3_i^i,inc2_i^j,zt2^i_i,zt2^j_i$ and does not receive a $tick$ message at time $k*\pval(\param)$, we are certain that at time $(k+1)*\pval(\param)$ its clock will be above $\pval(\param)$ and it will thus be stuck forever.
	\item Similarly if a process is in state $dec1_i^j,dec1_i^i,dec2^i_i,inc1_i^j,k^2$ and does not reset the clock when it is possible it will be stuck forever. 
	\item If an increment is requested by the controller part but the counter value is already equal to $\pval(\param)-1$ \ie{} the clock value of the counter process is equal to~$\pval(\param)$, then the processes are sent to an error state and thus for the rest of the execution there will not be any processes in the counter part.
	\item Similarly, if an increment is requested while no processes are left in the idle state, then the execution gets stuck in the next zero test.
\end{enumerate}

In other words, if a process does not behave correctly, its clock will increase over $\pval(\param)$ and the process will be stuck forever.
\begin{example}
Before going further, let us first give some example of the behavior of the network encoding the two-counter machine.
\begin{description}
\item[Successful decrement] $\TwoCMinstruction:decr~c_1~goto~\TwoCMinstruction_1$ with $v_2\geq v_1$ and $v_2+1 \leq \pval(\param)$ (those assumptions only matter for the order of the transitions).
\begin{scriptsize}
  \begin{align*}
%\scriptscriptstyle
  \left( \begin{array}{c}
k^1,0 \\
c_1,v_1 \\
c_2,v_2 
\end{array} \right)
\xrightarrow{1,1,!!dec_1,\{2,3\}} 
\left( \begin{array}{c}
k^2,1 \\
dec1_1^1,v_1+1 \\
dec1^2_1,v_2+1 
\end{array} \right)
\xrightarrow{\param-(v_2+1),3,\epsilon,\emptyset} 
\left( \begin{array}{c}
k^2,\param-v_2 \\
dec1_1^1,v_1+\param-v_2 \\
dec2_1^2,0 
\end{array} \right)
\\
%\param-(v_1+1+\param-(v_2+1))
%=\param-v_1-1-\param+v_2+1
%=-v_1+v_2
\xrightarrow{v_2-v_1,2,\epsilon,\emptyset} 
\left( \begin{array}{c}
k^2,\param-v_1 \\
dec2_1^1,0 \\
dec2_1^2,v_2-v_1
\end{array} \right)
\xrightarrow{1,2,\epsilon,\emptyset} 
\left( \begin{array}{c}
k^2,\param-v_1 +1\\
dec2_1^1,0 \\
dec2_1^2,v_2-v_1+1
\end{array} \right)
\xrightarrow{v_1-1,1,!!tick,\{2,3\}} 
\left( \begin{array}{c}
k^1_1,0\\
c_1,v_1-1 \\
c_2,v_2
\end{array} \right)
  \end{align*}
\end{scriptsize}
  
\item[Failed decrement] $\TwoCMinstruction:decr~c_1~goto~\TwoCMinstruction_1$ with $v_2\geq v_1$ and $v_2+1 \leq \pval(\param)$\LongVersion{ (those assumptions only matter for the order of the transitions)}.
\begin{scriptsize}
  \begin{align*}
%\scriptscriptstyle
\left( \begin{array}{c}
k^1,0 \\
c_1,0 \\
c_2,v_2 
\end{array} \right)
\xrightarrow{1,1,!!dec_1,\{2,3\}} 
\left( \begin{array}{c}
k^2,1 \\
c1_1,1 \\
dec1^2_1,v_2+1 
\end{array} \right)
\xrightarrow{\param-(v_2+1),3,\epsilon,\emptyset} 
\left( \begin{array}{c}
k^2,\param-v_2 \\
c_1,\param-v_2 \\
dec2_1^2,0 
\end{array} \right)
\end{align*}
\end{scriptsize}

Notice that for the rest of the execution the process $2$ will be stuck in $c_1$, unable to perform any action, nor to receive any message.
\LongVersion{We give two further examples (successful and failed increment) in \cref{appendix:proof:thm:reconf-PTBP-details:executions}.}
\end{description}
\end{example}

Let us now show by induction on~$k$ that either $\rho_{T=(k+1)*\pval(\param)+1/2}$ is not defined or $\rho_{T=k*\pval(\param)+1/2}$ encodes~$\TwoCMconfig_k$. 

The case $k=0$ is direct. By definition it is easy to see that $\rho_{T=1/2}$ encodes~$\config_0$.

Assume that the property holds for~$k$.
Let $\rho$ be an execution such that $\rho_{T=k*\pval(\param)+1/2}$ encode $\TwoCMconfig_k$. 
By the above remarks we have seen that if the network does not behave in the correct way it will get stuck before the next $\pval$ time unit thus $\rho_{T=(k+2)*\pval(\param)+1/2}$. 
The only thing left to show that the reduction is correct is that the clocks are reset at the right time to correctly model increment and decrement and that zero tests are correct.
For the latter, it is easy to see that by construction the controller part goes to the $k_z$ instruction if and only if its clock is equal to the counter clock hence the counter is equal to~$0$, otherwise it moves to $k_{nz}$.
For the former, the clocks evolve as in~\cite{JLR15}.
The only difference is for the increment where we need to introduce a new process used to guess when the clock value of the counter is equal to $\pval(\param)-1$.

We thus obtain that if the controller part can reach $k_{acc}$ then since the execution correctly encodes the run, the run must terminate.
Conversely if the run is infinite, for any $N$ and any $\pval$, any execution will either be infinite (and correct) thus never reaching $k_{acc}$, or eventually get stuck either because of an error in message, or because the counter clock is equal to $1$ during an increment, or because there will not be enough processes in the idle state.

This concludes the proof that $\exists$-EF-existence is undecidable for 1-clock PTBP in the reconfigurable semantics.

For $\exists$-EF-universality, notice that the error state $error$ is reachable only if an increment is requested when the counter value is equal to $\pval(\param)-1$. Thus if the error state is reached for all parameter valuations, this means that the run is unbounded.
Conversely if the run is unbounded for all parameter valuations, at some point the counter value is equal to $\pval(\param)-1$ during an increment and thus the error state is reachable.
To conclude on the undecidability of  $\exists$-EF-universality, we just have to recall that we consider rational valuations for the parameters, but in this proof we only used integer valuations.
This does not harm the proof of undecidability since we can modify the aforementioned gadget given in \LongVersion{\cref{fig:gadget}}\ShortVersion{\cite{ADFL18}} by replacing the state not integer by $error$.
This modification ensures that $error$ is reachable for any non integer valuation and the above argument that it is reachable for all integer valuations if and only if the two-counter machine is unbounded.\hfill \qed

\end{proof}

%%%%%%%%%%%%%%%%%%%%%%%%%%%%%%%%%%%%%%%%%%%%%%%%%%%%%%%%%%%%
%%%%%%%%%%%%%%%%%%%%%%%%%%%%%%%%%%%%%%%%%%%%%%%%%%%%%%%%%%%%
\section{Clique}\label{section:clique}
%%%%%%%%%%%%%%%%%%%%%%%%%%%%%%%%%%%%%%%%%%%%%%%%%%%%%%%%%%%%
%%%%%%%%%%%%%%%%%%%%%%%%%%%%%%%%%%%%%%%%%%%%%%%%%%%%%%%%%%%%

In broadcast protocol networks with a parametric number of processes, the topology of message communication plays a decisive role on the decidability status. In this section, we thus investigate a communication setting in which every message is received by all the other processes. We call these networks clique networks.

Formally, the semantics of a clique network is the restriction of the semantics given in~\cref{section:definitions} to internal transitions and broadcast transitions in which the set of receivers is always composed of all processes.

%%%%%%%%%%%%%%%%%%%%%%%%%%%%%%%%%%%%%%%%%%%%%%%%%%%%%%%%%%%%
\subsection{AF problems in the clique semantics}
%%%%%%%%%%%%%%%%%%%%%%%%%%%%%%%%%%%%%%%%%%%%%%%%%%%%%%%%%%%%

We first rule out the $\exists$-AF problem for the clique semantics, as we can show from~\cite{fournier2015parameterized} that it is undecidable already without any clock.
\ea{tu n'as pas une conf / revue ? ça me gêne un tout petit peu de citer une thèse (aussi sérieuse soit-elle)}
\pf{non :(, en fait on as jamais regardé AF, on fait toujours du EF mais dans ma these on considere des MDP et le probleme proba min de reach =1 . pour la preuve d'indecidabilité y'a pas de proba dans les MDP du coup ca reviens a faire AF. Mais on ne l'a jamais ecrit comme ca ailleur}

\ea{le théorème suivant marche évidemment quelles que soient les bornes}

\newcommand{\thmCliqueAFA}{The $\exists$-AF problem is \undec{undecidable} for PTBP with no clock in the clique semantics.}
%------------------------------------------------------------
\begin{theorem}\label{thm:clique-AFA}
	\thmCliqueAFA{}
\end{theorem}
%------------------------------------------------------------

\begin{proof}
	In \cite[Chapter III, Theorem~3.5]{fournier2015parameterized} it is shown that one can reduce the halting problem of a two-counter machine (which is undecidable \cite{minsky1967computation}) to the AF problem in a clique network without clocks.

	Intuitively the reduction goes as follows: the values of the counter are encoded by the number of processes in a given state. Increment and decrement of counter are easy to encode since in the clique semantics when one process sends a message everyone receives it, thus we can ensure that only one process performs the increment or decrement. The difficulty comes from the zero tests. Indeed, since we cannot force processes to answer we cannot differentiate between the case where there is no process encoding a counter and the case where the  processes do not answer. To tackle this problem, zero tests are implemented non-deterministically: if we choose that the counter is zero, a message is sent. If it was not the case, then the processes encoding the counter value move to an error state. In the case we choose that the value is not zero, the network is locked until a process encoding the counter sends a message or a process moves to the error state. This encoding ensures that every run that does not encode truthfully the two-counter machine reaches the error state. Thus by adding a transition from the halting state of the counter machine toward the error state, we can ensure that every path reaches the error state if and only if the two-counter machine halts. \hfill \qed
	\ea{à titre personnel, je virerais toute la preuve puisque c'est déjà prouvé, et qu'elle pas réutilisée (pour modification) ici}
\pf{on se ressert de l'idée pour la preuve de \cref{thm:clique-EFULUPTA} c'est pour ca que j'ai remis l'idee la}
\end{proof}
%------------------------------------------------------------

%%%%%%%%%%%%%%%%%%%%%%%%%%%%%%%%%%%%%%%%%%%%%%%%%%%%%%%%%%%%
\subsection{EF problems in the clique semantics}
%%%%%%%%%%%%%%%%%%%%%%%%%%%%%%%%%%%%%%%%%%%%%%%%%%%%%%%%%%%%

Recall that the proof of \cref{thm:reconf-PTBP} %(see \cref{proof:thm:reconf-PTBP})
	does not rely on the reconfigurable semantics particularity.
In fact the strong synchronization of processes in the clique semantics makes it even easier.
We thus obtain the following lemma:
%It was shown in~\cite{ADRST11} %\ea{y'a plein de réfs à ce papier, qui n'est pas défini dans le bibtex ; Paulin, veux-tu dire \cite{ADRST11} ?}
%that in this setting, without parametric guards, the EF problem is undecidable in general (already with 2 clocks) %\ea{donc à partir de combien d'horloges par process ? 2 ? il faut le préciser} 
%but decidable in the case where there is only one clock per process.
%In this section, we thus restrict our attention to PTBP with one clock.
%
%
%First, notice that considering a parametric number of processes is harder than considering a single PTA. 
%Indeed, we can see a PTA as a PTBP (with only internal action and without any broadcast edges)%\ea{là je vois pas trop… s'il n'y a pas de broadcast, comment assurer la synchro entre les morceaux de PTA ?}
%; then, solving the EF-universality/emptiness problem for the PTBP would amount to solve those problems for the initial PTA.
%Since EF-emptiness~\cite{AHV93} and EF-universality~\cite{ALR16ICFEM} are already undecidable for one clock PTA%\ea{mais pour une seule horloge, hein}
%, we obtain the following theorem. 

\ea{une fois de plus, ne fonctionne que pour des paramètres non bornés}

%------------------------------------------------------------
\begin{lemma}\label{lemma:clique-EFPTA}\label{lemma:clique-EFUPTA}
	The $\exists$-EF-existence and $\exists$-EF-universality problems are \undec{undecidable} for PTBP.
\end{lemma}
%------------------------------------------------------------

This undecidability does not hold in the case where each parameter appears either always as an upper bound or always as a lower bound in guards (but not both).
We thus consider in the following the case of L/U-PTBP.

%%%%%%%%%%%%%%%%%%%%%%%%%%%%%%%%%%%%%%%%%%%%%%%%%%%%%%%%%%%%
%%%%%%%%%%%%%%%%%%%%%%%%%%%%%%%%%%%%%%%%%%%%%%%%%%%%%%%%%%%%
\section{1-clock L/U-PTBP}\label{section:LU}
%%%%%%%%%%%%%%%%%%%%%%%%%%%%%%%%%%%%%%%%%%%%%%%%%%%%%%%%%%%%
%%%%%%%%%%%%%%%%%%%%%%%%%%%%%%%%%%%%%%%%%%%%%%%%%%%%%%%%%%%%

Since the L/U restriction brings some decidability to PTAs, we focus in this section on L/U-PTBP. Recall that L/U-PTA are expressive enough to model classical examples from the literature~\cite{HRSV02}, such as root contention or Fischer's mutual exclusion algorithm. As a consequence, L/U-PTBP make an interesting subclass of PTBP.

Due to the undecidability results of~\cite{ADRST11} for processes with 2 clocks (already without parameters), we consider in this section L/U-PTBP with one clock only.
When considering L/U-PTBP\LongVersion{ (PTBP where each parameter appear either as an upper bound in guards or as a lower bound but not both)}, we can get the following monotonicity result on the timing parameter valuations.

%\ea{pourquoi l'hypothèse d'une horloge ? ça ne marcherait pas avec 2 ou plus ? (en PTA si) Ah oui, c'est indécidable sans paramètres :/ Il faut peut-être le rappeler car, pour qui vient des PTA, c'est curieux.}

%------------------------------------------------------------
\begin{lemma}\label{lem:mono}
	Given $\PTBP$ an L/U-PTBP with one clock, a network size $N\in\mathbb{N}$, and a parameter valuation $\pval$,  for all valuations $\pval'$ such that for all upper-bound parameters~$\paramu$, $\pval(\paramu)\leq \pval'(\paramu)$ and for lower-bound parameters~$\paraml$, $\pval(\paraml)\geq \pval'(\paraml)$ we have that $\forall\rho\in \Execs(\PTBP,N,\pval)$, $\exists\rho'\in\Execs(\PTBP,N,\pval')$ such that $\rho$ is a prefix of $\rho'$.
\end{lemma}
%------------------------------------------------------------
\begin{proof}
	The proof is direct from the semantics definition. Notice that we do not have full inclusion of $\Execs(\PTBP,N,\pval)$ in $\Execs(\PTBP,N,\pval')$ since we consider maximal executions and it may be the case that some executions of $\Execs(\PTBP,N,\pval)$ appear only as prefixes of executions of $\Execs(\PTBP,N,\pval')$.
	Notice also that this holds in both semantics (reconfigurable and clique). \hfill \qed
\end{proof}
%------------------------------------------------------------

A direct consequence of \cref{lem:mono} and the decidability of the EF problem for PTBP with a single clock and without parameters is the decidability of the $\exists$-EF-existence problems for L/U-PTBP with one clock. 

\ea{je pense qu'on peut étendre le résultat ci-dessous à du open bounded : en jouant sur le graphe des régions qui passent par des bornes ouvertes et en prenant la valeur juste au-dessus du ``min'' (ah, sauf qu'en fait ça contredirait \cref{thm:clique-EFULUPTA}, hem hem)

Par contre, très clairement, le raisonnement ne marche plus en unbounded, puisque $x > \infty$ signifie désactiver la transition, et j'ai un contre-ex trivial où EF-univ est vrai pour un L/U mais l'état goal n'est pas accessible si on désactive la transition
}

\newcommand{\thmEFULUPTA}{The $\exists$-EF-universality problem is \dec{decidable} for closed bounded L/U-PTBP with one clock in both semantics.% when upper parameters have a closed lower bound and lower parameters have closed upper bound.\ea{ce qu'on peut appeler des ``closed L/U-PTBP'' selon notre terminologie dans \cite{AL17,ALR16ICFEM}}
}
%------------------------------------------------------------
\begin{lemma}\label{lemma:EFULUPTA}
	\thmEFULUPTA{}
\end{lemma}
%------------------------------------------------------------

\begin{proof}
	Let $\PTBP$ be an L/U-PTBP with one clock, and $\bounds$ be the closed bounds on the parameters.
	Let $\pval_{min}$ be the minimal permissive valuation \ie{} the valuation such that for all upper-bound parameters $\paramu$, $\pval_{min}(\paramu)=\boundinf(\paramu,\bounds)$ and for all lower-bound parameters $\paraml$, $\pval_{min}(\paraml)=\boundsup(\paraml,\bounds)$. %\ea{faudrait utiliser les notations dans la déf \ref{definition:boundedPTA} (qui n'est pas finalisée d'ailleurs); d'autant que les intervalles ont aussi le droit d'être fermés de l'autre côté…}        
	By definition we have $\pval_{min}\in \bounds$.%\ea{donc $b$ est à la fois les bornes et la borne sup…? :/ À mon avis, mieux vaut utiliser les notations du haut (avec des macros :p)}

	We define the PTBP without parameters $\PTBP_{min}$ as $\PTBP$ but replacing each occurrence of an upper-bound parameter $\paramu$ by $\boundinf(\paramu,\bounds)$ and each occurrence of a lower-bound parameter $\paraml$ by $\boundsup(\paraml,\bounds)$.	It is then easy to see that $\Execs(\PTBP,N,\pval_{min})=\Execs(\PTBP_{min},N)$. 

	Assume that for all $N$ there is no execution reaching $\locfinal$ in $\Execs(\PTBP_{min},N)$; then the above equality implies that the answer to $\exists$-EF-universality is false.

	Conversely assuming that there exists an execution reaching $\locfinal$ in $\Execs(\PTBP_{min},N)$ for some~$N$, we obtain by the equality and the monotonicity \cref{lem:mono} that this execution is a prefix of an execution of $\Execs(\PTBP,N,\pval)$ for any valuation~$\pval$.
	\ea{si on manque de place, je pense que ce raisonnement peut être largement condensé}

	Thus the $\exists$-EF-universality problem for $\PTBP$ is equivalent to the $\exists$-EF problem for $\PTBP_{min}$ and thus is decidable in the clique semantics (see \cite{ADRST11}) and in the reconfigurable semantics (see \cref{thm:reconf-EFTA}). \hfill \qed
\end{proof}
%------------------------------------------------------------

For the $\exists$-EF-existence problem, we can remove the assumption on the closed bounds.\ea{certes, mais dans la preuve tu as toujours besoin des bornes (sauf qu'elles peuvent être ouvertes) Le résultat devrait donc être pour closed bounded ou open bounded L/U, mais pour les (unbounded) L/U, ou bien ?}
\pf{Mouais, en fait, pour moi, les problemes ne sont definie qu'avec des bornes (sinon $\exists \pval\in\bounds$ n'a pas de sens) ce que je veut dire la c'est juste qu'on peut prendre n'importe quel forme de borne}\ea{oui mais il faut aussi considérer les unbounded, puisqu'on les a définis comme tels : je rajoute, du coup}

\ea{a priori ça doit marcher sans bornes, donc avec l'infini en sup, puisqu'on ne s'intéresse qu'à EF (à prouver, quand même ! car dans \cite{HRSV02}, ils ne font pas ça…)}

%------------------------------------------------------------
\begin{lemma}\label{lemma:EFLUPTA}
	The $\exists$-EF-existence problem is \dec{decidable} for (open or closed) bounded L/U-PTBP with one clock in both semantics.
\end{lemma}
%------------------------------------------------------------
\begin{proof}
	Let $\PTBP$ be an L/U-PTBP with one clock, and $\bounds$ be the bounds on the parameters.
	%\ea{notations…}
%
	As for the $\exists$-EF-universality problem, we define a protocol $\PTBP_{max}$ with the difference that non-strict guards involving open bounded parameters are changed to strict guards.
	We define the PTBP without parameters $\PTBP_{max}$ as $\PTBP$ but
for all upper-bound parameters $\paramu$ if $\bounds(\paramu)$ is of the form $(\boundinf,\boundsup]$ or $[\boundinf,\boundsup]$ then every occurrence of $\paramu$ is replaced by $\boundsup$. Otherwise if $\bounds(\paramu)$ is of the form $(\boundinf,\boundsup)$ or $[\boundinf,\boundsup)$ then every guard of the form $x<\paramu$ or $x\leq \paramu$ is replaced by the guard $x<\boundsup$. We operate similarly for lower-bound parameters.
%	 replacing each upper guard with an upper parameter $\paramu$ by $\clock<\boundinf(\paramu,\bounds)$ and for all lower guard with a lower parameter $\paraml$, by $\clock>\boundsup(\paraml,\bounds)$.
	% where $b(\paraml)=(a,b)$.\ea{notations}\ea{le problème, c'est qu'en enlevant l'hypothèse des bornes fermées, tu imposes en fait celles des bornes ouvertes ! or les L/U-PTA à bornes ouvertes et à bornes fermées sont incomparables. Je pense que c'est bon quand même, mais il faut combiner les deux raisonnements selon si chaque borne est ouverte ou fermée.}

	Using the same argument as for the monotonicity \cref{lem:mono} it is easy to see that for any valuation $\pval\in \bounds$,  any execution $\rho$ in $\Execs(\PTBP,N,\pval)$ is a prefix of some execution in $\Execs(\PTBP_{max},N)$.
	Thus if some execution reaches $\locfinal$ for some $N$ and some $p$ in $\Execs(\PTBP,N,\pval)$, there is also an execution reaching $\locfinal$ in $\Execs(\PTBP_{max},N)$.

	The other direction is more subtle. 
	Assume that there exists an execution $\rho$ reaching $\locfinal$ in $\Execs(\PTBP_{max},N)$. 
	Let $\rho'$ be a finite prefix of $\rho$ reaching $\locfinal$. 
	We define a valuation $\pval\in \bounds$ that contains an execution identical to $\rho'$ as follows:
	Let $\paramu$ be an upper-bound parameter. Either $\bounds(\paramu)$ is of the form $(\boundinf,\boundsup]$ or $[\boundinf,\boundsup]$ and we define $\pval(\paramu)=\boundsup$. Or
	$\bounds(\paramu)$ is of the form $(\boundinf,\boundsup)$ or $[\boundinf,\boundsup)$. In this case, let $v_u$ be the maximal value of clock $\clock$ along $\rho'$ when $\clock$ is compared in a guard which was formerly $\clock\compOp \paramu$. %\ea{là, je vois pas trop, comment peux-tu parler de $\rho'$ alors qu'il dépend de $\pval$… que tu es justement en train de définir ?}
	By definition of $\PTBP_{max}$ we know that $v_u<\boundsup$. We thus define $\pval(\paramu)=v_u+\epsilon$ with $\epsilon>0$, $\epsilon+v_u<\boundsup$ and $\epsilon>\boundinf-v_u$ (it exists since necessarily $\boundinf<\boundsup$).

	We operate in a symmetrical way for lower-bound parameters: $v_l$ is the minimal value of clock $\clock$ along $\rho'$ when $\clock$ is compared in a guard which was formerly $\clock\compOp \paraml$ and $\pval(u)=v_u-\epsilon$ with $v_u-\epsilon>\boundinf$, $\epsilon>0$ and $\epsilon<\boundsup+v_u$ (it exists since necessarily $\boundsup>\boundinf$).

	It is easy to see that for this valuation, $\rho'$ is a prefix of some execution in  $\Execs(\PTBP,N,\pval)$.
	Hence, the $\exists$-EF-existence problem for $\PTBP$ is equivalent to the EF problem for $\PTBP_{max}$ and thus decidable in the clique semantics (\cite{ADRST11}) and in the reconfigurable semantics (\cref{thm:reconf-EFTA}). \hfill \qed
\end{proof}
%------------------------------------------------------------

%\subsection{Clique semantics}

In contrast with the $\exists$-EF-existence problem, the monotonicity result is not enough to show decidability of the $\exists$-EF-universality problem for L/U-PTBP with open bounds.
In fact we can even show that the problem becomes undecidable for general L/U-PTBP in the clique semantics. 
More precisely it is undecidable for U-PTBP with one parameter with open left bound, and for L-PTBP with one unbounded parameter.

\reviewer{3 (FSTTCS)}{perhaps in para before the authors can explain why the earlier proof fails when you have open guards and then present this undecidability proof.}

\ea{a priori ce résultat est ok en open bounded ; mais pas en unbounded, puisque le cas de $\paramu = 0$ nous embête franchement}

\newcommand{\thmCliqueEFULUPTA}{The $\exists$-EF-universality problem is \undec{undecidable} for open bounded L/U-PTBP with one clock in the clique semantics.}

\begin{theorem}\label{thm:clique-EFULUPTA}
	\thmCliqueEFULUPTA{}
\end{theorem}

% BEGIN VERSION SKETCH
% \begin{proof}[Proof sketch]
% The idea is to encode a two-counter machine, the number of processes in a particular state is used to encode the counter value.
% Thanks to the clique semantics, increment and decrement of counter is easy to simulate.
% However, zero tests are not possible since there is no way to distinguish between the fact that no process is modeling a counter and the fact that they just do not send a message.
% We thus allow the simulation to guess whether the counter is zero or not zero non-deterministically, in case of a wrong guess we are able to detect it thanks to the clique semantics. In this case at least one process is stuck in an error state, we then use the time parameter to repeat the simulation an unbounded number of time before moving to the target state. To be able to reach the target state we thus have to be able to correctly simulate the two-counter machine without wrong guess.
% 
% Full details are given in \cref{proof:thm:clique-EFULUPTA}.
% \end{proof}
% END VERSION SKETCH

\begin{proof}
We reduce from the halting problem of two-counter machines.
The idea is to encode a two-counter machine, the number of processes in a particular state is used to encode the counter value.
Thanks to the clique semantics, increment and decrement of counters are easy to simulate.
However, zero tests are not possible since there is no way to distinguish between the fact that no process is modeling a counter and the fact that they just do not send a message.
We thus allow the simulation to guess whether the counter is zero or not zero non-deterministically; in case of a wrong guess we are able to detect it thanks to the clique semantics.
In this case, at least one process is stuck in an error state, we then use the timing parameter to repeat the simulation an unbounded (but finite) number of times before moving to the target state. To be able to reach the target state, we thus have to be able to correctly simulate the two-counter machine without wrong guess.

% We reuse elements of the proof of \cite[Chapter~III, Theorem~3.5]{fournier2015parameterized}.
Formally, given a two-counter machine $\TWOCM = (\TwoCMInstructions,\TwoCMinstructionZero,\TwoCMinstructionAcc)$ we define a PTBP $\mathcal{P}$ as follows:
\begin{itemize}
	\item $\Loc=\{\locinit,idle,c_i,c_i^d,c_i^i,c_i^z,err,\locfinal\mid i\in\{1,2\}\}\cup \{k,k'\mid k\in K\}$ where, $\locinit$ is the initial state, $idle$ is a waiting state for the processes encoding the counters, $c_i$ is the state used to encode the value of counter~$\TwoCMCounter_i$, $c_i^i$ and $c_i^d$ are intermediary states for increment and decrement of counter $c_i$, $c_i^z$ is an intermediary state used for the zero test, a state $k$ is used to encode that the simulation reached instruction $k$ of the machine and $k'$ is an intermediary state, $err$ is a sink state used to detect error in the simulation, finally $\locfinal$ is the target state.
	\item $\Clock = \{ \clock \}$ and $\Param=\{\paramu,\paraml\}$
	\item $\Alphabet=\{inc_i,dec_i,z_i,nz_i,ok,end\mid i\in\{1,2\}\}$ where $inc_i$, $dec_i$, $z_i$, and $nz_i$ stand respectively for increment, decrement, zero, and not zero of counter~$c_i$, $ok$ is a message to acknowledge that the action was performed correctly, and $end$ is the message sent at the end of the simulation to either restart a simulation or reach the target state.
	\item $\Edges$ is defined as follows (for simplicity the guard and update of the clock are omitted when trivial, \ie{} the true guard and no reset):
	\begin{description}
	\item[Initialization.] $(\locinit,!!ok,k_0)\in \Edges$, $(\locinit,??ok,idle)\in\Edges$.
	\item[Increment of counter $i$.] For an increment instruction $\TwoCMinstruction: incr~\TwoCMCounter_i~goto~\TwoCMinstruction_1$, we add to $\Edges$ the transitions: $(k,!!inc_i,k')$, $(k',??ok,k_1)$, $(idle,??inc_i,c_i^i)$, $(c_i^i,!!ok,c_i)$ $(c_i^i,??ok,idle)$.
	\item[Decrement of counter $i$.] For a decrement instruction $\TwoCMinstruction: decr~c_i~goto~\TwoCMinstruction_1$, we add to $\Edges$ the transitions: $(k,!!dec_i,k')$, $(k',??ok,k_1)$, $(c_i,??dec_i,c_i^d)$, $(c_i^d,!!ok,idle)$ $(c_i^d,??ok,c_i)$.
	\item[Zero-test of counter $i$.] For a zero-test instruction $\TwoCMinstruction: if~c_i=0~goto~\TwoCMinstruction_z~\linebreak[0]else~goto~\TwoCMinstruction_{nz}$, we add to $\Edges$ the transitions: $(k,!!z_i,k_z)$, $(k,!!nz_i,k')$, $(k',??ok,k_{nz})$, $(c_i,??z_i,err)$, $(c_i,??nz_i,c_i^z)$ $(c_i^z,!!ok,c_i)$, $(c_i^z,??ok,c_i)$.
	\item[End of simulation.]\hspace{1cm}\\ $(k_{acc},\clock<\paramu,!!end,\{\clock:=0\},k_0)$ $(idle,\clock>\paraml,??end,\locfinal)$ $(c_i,??end,idle)$.
	\end{description}
	\end{itemize}
	Given a configuration $\config$ of the network, we say that it encodes a configuration $(k,v_1,v_2)$ of the two-counter machine if there is one process in state $k$ and $v_i$ processes in states $c_i$ for $i\in\{1,2\}$.
	If we omit the \emph{end of simulation} part, this reduction is similar to the one found in  \cite[Chapter~III, Theorem~3.5]{fournier2015parameterized};
	we therefore proceed with less details on this part.
	In short, every execution of the network is of one of the three kinds:
% 		we leave the reader find there the detailed proof\ea{audacieux !} of the fact that every execution of the network is of one of the three kinds:
	\begin{description}
		\item[Correct simulation.] The execution correctly encodes the run of the two-counter machine.
		\item[Lack of processes.] The controller is stuck in an intermediary state while performing an increment, \emph{i.e} there was no process left in the $idle$ state when the controller sent the $inc_i$ message, thus it is stuck waiting for an $ok$ message that no one can send.
		\item[Wrong zero-test.] Along the execution, the controller wrongly assumed the value of a counter. Either it guessed a non-zero value and it is stuck waiting for an $ok$ message, or it guessed zero when it was not---in which case at least one process moved to the error state.
\end{description}

\todo{CONCUR 2018: ``Even if the exact complexity is unknown, you should mention any known upper and lower bounds. ''}

Notice now that to reach the target state $\locfinal$ a process in $idle$ must receive the message $end$ after its clock value is greater than parameter $\paraml$.
But the end of simulation part requires that the controller clock is lower than parameter~$\paramu$.
Thus when reaching state $k_{acc}$, in order to be able to let more time elapse, the controller has to send the message $end$ which leads to a configuration where there is no process in the counter states and the controller is in the initial state of the two-counter machine. This configuration thus encodes the initial configuration of the two-counter machine. The controller then must simulate another time the two-counter machine before being able to send $end$ again.

Thus, given a valuation $\pval$ of the parameters, to reach $\locfinal$ at least $\pval(\paraml)/\pval(\paramu)$ messages $end$ must be sent by the controller. In other words, $\pval(\paraml)/\pval(\paramu)$ (correct or incorrect) simulations of the two-counter machine must be performed before reaching $\locfinal$. We have seen before that every incorrect simulation either gets stuck, or sends at least a process in the error state. Hence, given a network size~$N$, if for a valuation $\pval$ such that $\pval(\paraml)/\pval(\paramu)>N$ the state $\locfinal$ is reached, then at least one simulation was correct, thus the two-counter machine halts.

This proves the undecidability of the EF-universality problem with~0 as an open lower bound for $\paramu$.
Indeed, if there exists a network of size~$N$ which satisfies the EF-universality, then it is possible to reach $\locfinal$ for all valuation and in particular for a valuation such that $\pval(\paraml)/\pval(\paramu)>N$.
For the other direction, if the machine halts, there exists a size of network ($m+2$ where $m$ is the maximal sum of the two-counter value along the execution) that ensures that $\locfinal$ is reachable for any valuation $\pval$ with $\pval(\paramu)>0$.
Indeed, the controller can simulate the two-counter machine correctly (since it has enough processes to model the counters) in 0 time unit, wait a positive delay but less than $\paramu$ time unit, and repeat this until the clock value of the processes in $idle$ is greater than $\paraml$.
This is possible since every time the controller sends the message $end$ the configuration obtained is the same as the one obtained after the initialization (the first message $ok$). \hfill \qed
\end{proof}

\newcommand{\thmEFunivonelereste}{%
	$\exists$-EF-universality in the clique semantics is \undec{undecidable} already with a single clock for %L/U-PTBP,
		U-PTBP with open bounds on the left, and  L-PTBP with infinity as right bound.
	%	\begin{itemize}
	%		\item L/U-PTBP,
	%		\item U-PTBP with open bounds (on the left), and 
	%		\item L-PTBP with infinity as right bound
	%\end{itemize}
	%already with a single clock.
}

\begin{lemma}\label{EFuniv:1:lereste}
	\thmEFunivonelereste{}
\end{lemma}
%------------------------------------------------------------
\begin{proof}
	The proof of \cref{thm:clique-EFULUPTA} uses an open bounded L/U-PTBP.
	Moreover we only used the fact for all size of network $N$ there exists a valuation of the parameter $p$ such that $\pval(\paraml)/\pval(\paramu)>N$.
	Thus the proof can be adapted with only one upper-bound parameter $\paramu$ (resp.\ lower-bound parameter $\paraml$) by replacing $\paraml$ by $1$ in the protocol (resp.~$\paramu$ by~$1$).
	This still ensures that there exists a valuation such that $1/\pval(\paramu)>N$ (resp.\ $\pval(\paraml)>N$).  \hfill \qed
\end{proof}
%------------------------------------------------------------

\ea{il reste mine de rien quelques cas ouverts selon la façon dont les bornes sont ouvertes ou fermées, notamment :
\begin{itemize}
	\item unbounded L/U
	\item unbounded U
	\item open bounded L
\end{itemize}

}

%%%%%%%%%%%%%%%%%%%%%%%%%%%%%%%%%%%%%%%%%%%%%%%%%%%%%%%%%%%%
%%%%%%%%%%%%%%%%%%%%%%%%%%%%%%%%%%%%%%%%%%%%%%%%%%%%%%%%%%%%
\section{Conclusion}\label{section:conclusion}
%%%%%%%%%%%%%%%%%%%%%%%%%%%%%%%%%%%%%%%%%%%%%%%%%%%%%%%%%%%%
%%%%%%%%%%%%%%%%%%%%%%%%%%%%%%%%%%%%%%%%%%%%%%%%%%%%%%%%%%%%

% \newcommand\crefabbr[1]{%
% \begingroup
% 	\crefname{theorem}{\text{T}}{\text{T}}
% 	\crefname{lemma}{\text{L}}{\text{L}}
% 	\cref{#1}
% \endgroup%
% }

\newcommand\crefThmabbr[1]{Th\ref{#1}}
\newcommand\crefLemmaabbr[1]{L\ref{#1}}

\begin{table}[t]
	\footnotesize
	\setlength{\tabcolsep}{1pt} % . The default value is 6pt.
	%%%%%%%%%%%%%%%%%%%%%%%%%%%%%%%%%%%%%%%%%%%%%%%%%%%%%%%%%%%%
	\begin{subfigure}[b]{0.59\textwidth}
		\begin{tabular}{|l|c|c|c|c|c|c|c|c|}\hline
			% & 2 clocks
			& 1-c
			& 2-c
			& 3-c
			& \multicolumn{2}{c|}{1-L/U}
			& 2-L/U
			& 3-L/U
			\\
			%------------------------------------------------------------
			& 
			& 
			& 
			&cb 
			&ob
			%&u
			&
			&
			\\ \hline
			%------------------------------------------------------------
			$\exists$-EF-empt.
			& \multicolumn{3}{c|}{\cellundec{\crefThmabbr{thm:reconf-PTBP}}}
			& \multicolumn{2}{c|}{\celldec{\crefLemmaabbr{lemma:EFLUPTA}}}
			%& \cellundec{\crefThmabbr{thm:reconf-EFTA}}
			%& \cellopen{open}
			& \multicolumn{2}{c|}{\cellundec{\crefLemmaabbr{lemma:EF:undec:2}}}
			\\ \hline
			%------------------------------------------------------------
			$\exists$-EF-univ.
			& \multicolumn{3}{c|}{\cellundec{\crefThmabbr{thm:reconf-PTBP}}}
			& \celldec{\crefLemmaabbr{lemma:EFULUPTA}} %& \cellundec{\crefThmabbr{thm:reconf-EFTA}}
			& \multicolumn{1}{c|}{\cellopen{open}}
			& \multicolumn{2}{c|}{\cellundec{\crefLemmaabbr{lemma:EF:undec:2}}}
			\\ \hline
			%------------------------------------------------------------
			$\exists$-AF
			& \celldec{\crefThmabbr{thm:reconf-AF}}
			& \cellopen{open}
			& \cellundec{\crefThmabbr{thm:reconf-AF}}
			& \multicolumn{2}{c|}{\celldec{\crefThmabbr{thm:reconf-AF}}}
			& \cellopen{open}
			& \cellundec{\crefThmabbr{thm:reconf-AF}}
			\\\hline
			%------------------------------------------------------------
		\end{tabular}
		\caption{Reconfigurable semantics}
		\label{table:summary:reconfigurable}
	\end{subfigure}
	%%%%%%%%%%%%%%%%%%%%%%%%%%%%%%%%%%%%%%%%%%%%%%%%%%%%%%%%%%%%
	\begin{subfigure}[b]{0.45\textwidth}
		\begin{tabular}{|l|c|c|c|c|c|c|c|c|c|}\hline
% 				& A & 1-TA & 2-TA
			& PTBP
			& \multicolumn{2}{c|}{L/U}
			& \multicolumn{2}{c|}{L or U}
% 				& U
% 				& L
			\\
			%------------------------------------------------------------
			& 
			&cb 
			&ob
			%&u
			&cb 
			&ob
			%&u
				\\ \hline
				%------------------------------------------------------------
			$\exists$-EF-empt.
% 			& A & \celldec{\cite{ADRST11}}&\cellundec{\cite{ADRST11}}
			& \cellundec{\crefLemmaabbr{lemma:clique-EFPTA}}
			& \multicolumn{2}{c|}{\celldec{\crefLemmaabbr{lemma:EFLUPTA}} }
			%& \cellopen{open}
			& \multicolumn{2}{c|}{\celldec{\crefLemmaabbr{lemma:EFLUPTA}} }
			%& \cellopen{open}
			%------------------------------------------------------------
			\\ \hline
			$\exists$-EF-univ.
% 			& \celldec{D} & \celldec{\cite{ADRST11}} & \cellundec{\cite{ADRST11}}
			& \cellundec{\crefLemmaabbr{lemma:clique-EFUPTA}}
			& \celldec{\crefLemmaabbr{lemma:EFULUPTA}}
			& \cellundec{\crefThmabbr{thm:clique-EFULUPTA}}
			& \celldec{\crefLemmaabbr{lemma:EFULUPTA}}
			& \cellundec{\crefLemmaabbr{EFuniv:1:lereste}}
			%& \cellundec{\crefabbr{EFuniv:1:lereste}}
			%------------------------------------------------------------
			\\ \hline
			$\exists$-AF
			& \multicolumn{5}{c|}{\cellundec{\crefThmabbr{thm:clique-AFA}}}
			%------------------------------------------------------------
			\\ \hline
% 			AF-universality	& \multicolumn{7}{c|}{\cellundec{\ref{thm:clique-AFA}}}\\\hline
			
		\end{tabular}
		\caption{Clique semantics for 1 clock} % ($\geq$ 2 is undecidable)}
		\label{table:summary:clique}
	\end{subfigure}
	%%%%%%%%%%%%%%%%%%%%%%%%%%%%%%%%%%%%%%%%%%%%%%%%%%%%%%%%%%%%
	\caption{Summary of our contributions (bold green: decidable; red italic: undecidable)}
	\label{table:summary}
\end{table}

Up to our knowledge this work is the first to consider two different sets of parameters at the same time.
Both parameterized number of processes and parametric clocks are difficult to deal with and number of problems are undecidable for each of these systems.
However we have shown that the combination of the decidable subclasses leads to some decidable problems.
Our contributions are summarized in \cref{table:summary}; $i$-c (resp.~$i$-L/U) denotes PTBP (resp.\ L/U-PTBP) with $i$ clocks per process.
In \cref{table:summary:clique}, cb and ob denote formalisms with a closed bounded parameter domain and an open bounded parameter domain. % and an unbounded domain.

The open 2-clock case in the reconfigurable semantics is a well-known open problem, with connections to open problems of logic and automata theory~\cite{AHV93}.
The other open case in \cref{table:summary} we are interested in solving is $\exists$-EF-universality for 1-L/U-PTBP in the reconfigurable semantics with open bounds.
In addition, EF problems are still open for bounded 1-clock PTBP (\cref{thm:reconf-PTBP} requires unbounded parameters), and for 1-c L/U with unbounded parameters in the clique semantics.
Finally, for the decidable subclasses we exhibited, it remains to be studied whether exact synthesis can be achieved, \ie{} obtaining the set of sizes of processes and timing parameter valuations for which EF or AF holds.

Another future work is to consider the EF- and AF-emptiness problems; recall that, in contrast to formalisms with a network of size~1 (\ie{} PTAs), the emptiness problem is not equivalent to the existence problem, due to the additional quantifier over the network size.

More general future works include considering other semantics such as asynchronous broadcast or different communication topologies (reconfigurable under constraint, restricted to graph of bounded width, \dots), as well as the reachability problem for \emph{all} sizes of networks (instead of the \emph{existence} of a network size):
	while it seems straightforward for EF problems, it remains to be done for AF problems.
Another quantifier of interest is the number of processes reaching the target: so far, we considered the existence of one process reaching the target.
All processes reaching the target is also of interest.

% \subsection{Summary of the result for clique semantics}
% 
% 
% \begin{tabular}{|c|c|c|c|c|c|c|c|}\hline
% 		&Reconfigurable & Clique
% 		\\
% 		\hline
% 		%------------------------------------------------------------
% 		& 1-c & 2-c & $n$-c & L/U & U & L
% 		\\
% 		\hline
% 		%------------------------------------------------------------
% 		$\exists$-EF-empt.	& \celldec{\cite{ADRST11}}&\cellundec{\cite{ADRST11}} & \cellundec{\ref{lemma:clique-EFPTA}} & \multicolumn{3}{c|}{\celldec{\ref{lemma:EFLUPTA}} }
% 		\\
% 		\hline
% 		%------------------------------------------------------------
% 		$\exists$-EF-univ.	& \celldec{\cite{ADRST11}} & \cellundec{\cite{ADRST11}} & \cellundec{\ref{lemma:clique-EFUPTA}} & \cellundec{\ref{thm:clique-EFULUPTA}} & \cellundec{\ref{thm:clique-EFULUPTA}}\hspace{1em}\celldec{\ref{lemma:EFULUPTA}} & \cellundec{\ref{thm:clique-EFULUPTA}}\hspace{1em}\celldec{\ref{lemma:EFULUPTA}}
% 		\\
% 		\hline
% 		%------------------------------------------------------------
% 		$\exists$-AF	& \multicolumn{6}{c|}{\cellundec{\ref{thm:clique-AFA}}}
% 		\\
% 		\hline
% 		%------------------------------------------------------------
% 	
% \end{tabular}

\medskip
\textbf{Acknowledgement} The authors warmly thank Nathalie Bertrand for fruitful discussions on the topic of this paper.

%\clearpage

%%%%%%%%%%%%%%%%%%%%%%%%%%%%%%%%%%%%%%%%%%%%%%%%%%%%%%%%%%%%
\newcommand{\JLAP}{Journal of Logic and Algebraic Programming}
\newcommand{\LNCS}{Lecture Notes in Computer Science}
\newcommand{\STTT}{International Journal on Software Tools for Technology Transfer}
\newcommand{\TCS}{Theoretical Computer Science}
\newcommand{\IEEETSE}{IEEE Transactions on Software Engineering}

\ifdefined\VersionAuthor
	\bibliographystyle{alpha} % plain
\else
\bibliographystyle{splncs04} % abbrv
% % 	\newcommand{\JLAP}{JLAP}
% 	\newcommand{\JLAP}{Journal of Logic and Algebraic Programming}
% 	\newcommand{\LNCS}{LNCS}
% % 	\newcommand{\LNCS}{Lecture Notes in Computer Science}
% 	\newcommand{\STTT}{STTT}
% % 	\newcommand{\STTT}{International Journal on Software Tools for Technology Transfer}
% 	\newcommand{\TCS}{TCS}
% 	\newcommand{\IEEETSE}{IEEE TSE}
% % 	\newcommand{\IEEETSE}{IEEE Transactions on Software Engineering}
\fi
\bibliography{bib}

%%%%%%%%%%%%%%%%%%%%%%%%%%%%%%%%%%%%%%%%%%%%%%%%%%%%%%%%%%%%
%% .. or use the thebibliography environment explicitely

%%%%%%%%%%%%%%%%%%%%%%%%%%%%%%%%%%%%%%%%%%%%%%%%%%%%%%%%%%%%
%%%%%%%%%%%%%%%%%%%%%%%%%%%%%%%%%%%%%%%%%%%%%%%%%%%%%%%%%%%%
%%%%%%%%%%%%%%%%%%%%%%%%%%%%%%%%%%%%%%%%%%%%%%%%%%%%%%%%%%%%
% BEGIN APPENDIX IN LONG VERSION ONLY
\LongVersion{
\newpage
\appendix

%%%%%%%%%%%%%%%%%%%%%%%%%%%%%%%%%%%%%%%%%%%%%%%%%%%%%%%%%%%
%%%%%%%%%%%%%%%%%%%%%%%%%%%%%%%%%%%%%%%%%%%%%%%%%%%%%%%%%%%%
\section{Proof of \cref{thm:reconf-AF}}\label{proof:thm:reconf-AF}
%%%%%%%%%%%%%%%%%%%%%%%%%%%%%%%%%%%%%%%%%%%%%%%%%%%%%%%%%%%%
%%%%%%%%%%%%%%%%%%%%%%%%%%%%%%%%%%%%%%%%%%%%%%%%%%%%%%%%%%%%

\recallResult{thm:reconf-AF}{\thmReconfAF{}}

\begin{proof}
	We first show that in the reconfigurable semantics the $\exists$-AF problems are equivalent to the same problems but in networks of size one.

	The easiest direction is to assume that an AF property holds in a network of size one, it thus answers the original parameterized question which asks whether there exists a size of network satisfying the property.
	The other direction is more subtle and derives from the fact that in the reconfigurable semantics nothing prevents messages to be sent to an empty set of receivers.
	In the case of executions where there is no communication (every message is sent with an empty set of receivers) the network behaves as several processes running in parallel without interaction. Thus, if the AF problem is satisfied for some size of networks, it means that the target state is reached in particular for every execution without communication. Thus, it follows that the target is reached for all executions with a single process hence it holds for a network of size 1.

	The rest of the theorem follows from the literature.
	AF-emptiness and AF-universality are decidable for 1-PTA. Indeed it is shown in~\cite{AM15} that one can abstract 1-PTA semantics into a finite parameterized zone abstraction. Moreover one can use this abstraction to solve the AF problems as shown in  \cite{JLR15}.
	The undecidable cases directly come from the fact that the AF-emptiness is undecidable for (L/U-)PTA with 3 clocks or more~\cite{JLR15}.
	The decidable case comes from the following result for PTA with 1 clock, which was, to the best of our knowledge, never shown formally:

	\ea{AF-emptiness and universality is decidable for TA and PTA
	donc \dec{decidable} pour TA et PTA1 (pour ce dernier cas,), \undec{undec} pour PTA3 \cite{JLR15} et L/U-PTA, et ouvert pour L-PTA, U-PTA et PTA-2.}

	\begin{lemma}[Decidability of AF-emptiness for 1-clock PTA]
		The AF-emptiness problem is \dec{decidable} for 1-clock PTA.
	\end{lemma}
	\begin{proof}
	\pf{Bof}
	A classical way to approach timed automaton is to abstract the reachable configurations into a finite region graph (see \emph{e.g.} \cite{baier2008principles}). This region graph can be adapted in the case of PTA in order to deal with parameters. It was shown in \cite{AM15} that a similar abstraction called zone graph is finite for 1-clock PTA. To obtain the decidability of the AF-emptiness problem it then suffices to apply the symbolic algorithm given in \cite{JLR15}. \hfill \qed
	\end{proof} 
	
	This concludes the proof of \cref{thm:reconf-AF}. \hfill \qed

	% 	The undecidability result can also be found in~\cite{JLR15} where it is shown that the AF problems are undecidable for 3-PTA and L/U-PTA.

\end{proof} 
	Notice that these questions are still open for L-PTA, U-PTA and 2-PTA.
%------------------------------------------------------------

%------------------------------------------------------------
\begin{remark}\label{remark:reconf-AFopen}
  Following the reasoning in the proof of \cref{thm:reconf-AF}, the $\exists$-AF-emptiness problem is open in the same cases as for PTA:
	for L-PTBP and U-PTBP, and for (L/U)-PTBP with 2 clocks.
\end{remark}
\section{Two-counter machines}\label{appendix:2CM}
%%%%%%%%%%%%%%%%%%%%%%%%%%%%%%%%%%%%%%%%%%%%%%%%%%%%%%%%%%%%
%%%%%%%%%%%%%%%%%%%%%%%%%%%%%%%%%%%%%%%%%%%%%%%%%%%%%%%%%%%%
For completeness, we recall here the definition of two-counter machines as well as the undecidability of the halting problem. 

\begin{definition}[\cite{minsky1967computation}]
	A two-counter machine is a tuple $\TWOCM = (\TwoCMInstructions,\TwoCMinstructionZero,\TwoCMinstructionAcc)$ manipulating integer variables $\TwoCMCounter_1$ and $\TwoCMCounter_2$ called counters and composed of a finite set of instructions~$\TwoCMInstructions$.
	Each instruction $\TwoCMinstruction\in \TwoCMInstructions$ is either of the form:
	\begin{description}
			\item[Increment] $\TwoCMinstruction:inc~\TwoCMCounter_i;~goto~\TwoCMinstruction'$, or
			\item[Decrement] $\TwoCMinstruction:decr~\TwoCMCounter_i;~goto~\TwoCMinstruction'$, or
			\item[Zero test] $\TwoCMinstruction:if~\TwoCMCounter_i=0~goto~\TwoCMinstruction'~else~goto~\TwoCMinstruction''$
	\end{description}
	where $i\in\{1,2\}$ and $\TwoCMinstruction,\TwoCMinstruction',\TwoCMinstruction''$ are labels preceding instructions. $\TwoCMinstructionZero$ is the initial label and $\TwoCMinstructionAcc$ is the accepting label.
\end{definition}
A configuration is a tuple of $\TwoCMInstructions\times\grandn\times\grandn$.
A configuration $(\TwoCMinstruction,\TwoCMCounterValue_1,\TwoCMCounterValue_2)$ means that the machine is at instruction $\TwoCMinstruction$ with counter $\TwoCMCounter_1$ with value $\TwoCMCounterValue_1$ and counter $\TwoCMCounter_2$ with value $\TwoCMCounterValue_2$.

A two-counter machine gives rise to a run $\TwoCMconfig_0\to\TwoCMconfig_1\to\dots$ where $\TwoCMconfig_0=(\TwoCMinstruction,0,0)$ and for all $\TwoCMconfig_i=(\TwoCMinstruction,\TwoCMCounterValue_1,\TwoCMCounterValue_2)$ the successor configuration depends on the form of~$\TwoCMinstruction$.
\begin{itemize}
	\item if $\TwoCMinstruction:inc~\TwoCMCounter_1;~goto~\TwoCMinstruction'$ then $\TwoCMconfig_{i+1}=(\TwoCMinstruction',\TwoCMCounterValue_1+1,\TwoCMCounterValue_2)$ (similarly for an increment of $\TwoCMCounter_2$)
	\item if $\TwoCMinstruction:decr~\TwoCMCounter_1;~goto~\TwoCMinstruction'$ then $\TwoCMconfig_{i+1}=(\TwoCMinstruction',\TwoCMCounterValue_1-1,\TwoCMCounterValue_2)$ (similarly for an decrement of $\TwoCMCounter_2$)
	\item if $\TwoCMinstruction:if~\TwoCMCounter_1=0~goto~\TwoCMinstruction'~else~goto~\TwoCMinstruction''$  and $\TwoCMCounterValue_1=0$ then $\TwoCMconfig_{i+1}=(\TwoCMinstruction',0,\TwoCMCounterValue_2)$ (similarly for $\TwoCMCounter_2$)
	\item $\TwoCMinstruction:if~\TwoCMCounter_1 = 0~goto~\TwoCMinstruction'~else~goto~\TwoCMinstruction''$  and $\TwoCMCounterValue_1>0$ then $\TwoCMconfig_{i+1}=(\TwoCMinstruction'',\TwoCMCounterValue_1,\TwoCMCounterValue_2)$ (similarly for $\TwoCMCounter_2$)
\end{itemize}
 
We assume without loss of generality that the run is either infinite or the machine halt in~$\TwoCMinstructionAcc$. Moreover we can assume without loss of generality that the two last instructions are necessarily zero test of both counter.
The halting problem asks whether the machine halts, and the boundedness problem asks whether the counter values are bounded along the run.
Both problems have been shown undecidable in~\cite{minsky1967computation}.

%%%%%%%%%%%%%%%%%%%%%%%%%%%%%%%%%%%%%%%%%%%%%%%%%%%%%%%%%%%%
%%%%%%%%%%%%%%%%%%%%%%%%%%%%%%%%%%%%%%%%%%%%%%%%%%%%%%%%%%%%
\section{Additional details on the proof of \cref{thm:reconf-PTBP}}\label{appendix:proof:thm:reconf-PTBP-details}
%%%%%%%%%%%%%%%%%%%%%%%%%%%%%%%%%%%%%%%%%%%%%%%%%%%%%%%%%%%%
%%%%%%%%%%%%%%%%%%%%%%%%%%%%%%%%%%%%%%%%%%%%%%%%%%%%%%%%%%%%

\recallResult{thm:reconf-PTBP}{\thmReconfPTBP{}}

%%%%%%%%%%%%%%%%%%%%%%%%%%%%%%%%%%%%%%%%%%%%%%%%%%%%%%%%%%%%
\subsection{Gadget constraining integer valuations}\label{appendix:proof:thm:reconf-PTBP-details:integer}
%%%%%%%%%%%%%%%%%%%%%%%%%%%%%%%%%%%%%%%%%%%%%%%%%%%%%%%%%%%%

The gadget is given in \cref{fig:gadget}.

\begin{figure}[h!]
\begin{tikzpicture}[node distance=3cm]
\node[draw,circle] (0) {$0$};
\node[draw,circle] (1)[above right of=0] {$1$};
\node[draw,circle] (2)[right of=0] {$2$};
\node[draw,circle] (q0)[right of=2] {$q_0$};
\node (ni)[below right of=2] {not integer};

\path[->,above]
(0) edge node {$\epsilon$} (1)
(0) edge node {$\epsilon$} (2)
(1) edge [loop right] node {$x=\param,!!now$} (1)
(2) edge[loop above] node {$x=1,\epsilon,\{x:=0\}$} (2)
(2) edge node {$x=0,??now$} (q0)
(2) edge node [right]{$x>0\wedge x<1,??now$} (ni)

;

\end{tikzpicture}
\caption{Gadget to enforce integer values of $\param$}
\label{fig:gadget}
\end{figure}
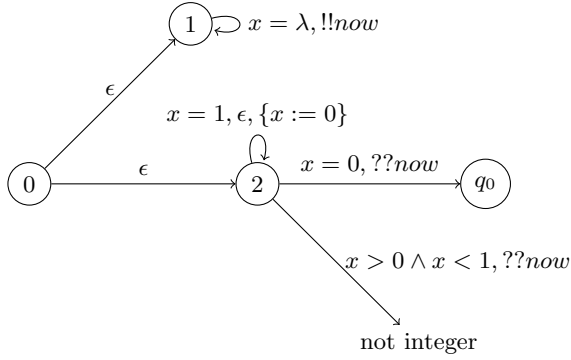

%%%%%%%%%%%%%%%%%%%%%%%%%%%%%%%%%%%%%%%%%%%%%%%%%%%%%%%%%%%%
\subsection{Additional examples of executions}\label{appendix:proof:thm:reconf-PTBP-details:executions}
%%%%%%%%%%%%%%%%%%%%%%%%%%%%%%%%%%%%%%%%%%%%%%%%%%%%%%%%%%%%

\begin{description}
	\item[Successful increment] $\TwoCMinstruction:incr~c_1~goto~\TwoCMinstruction_1$ with $v_1<\param-1$ to ensure that the reception can be taken, and $v_2\geq v_1$ and $v_2+1\leq\param$\LongVersion{ (those assumptions only matter for the order of the transitions)}.
	\begin{align*}
	\scriptscriptstyle
	\left( \begin{array}{c}
	k^1,0 \\
	c_1,v_1 \\
	c_2,v_2 \\
	idle,w
	\end{array} \right)
	\xrightarrow{1,1,!!inc_1,\{2,3\}} 
	\left( \begin{array}{c}
	k^2,1 \\
	inc1_1^1,v_1+1 \\
	inc1^2_1,v_2+1 \\
	idle,w'
	\end{array} \right)
	\xrightarrow{\param-(v_2+1),3,\epsilon,\emptyset} 
	\left( \begin{array}{c}
	k^2,\param-v_2 \\
	inc1_1^1,v_1+\param-v_2 \\
	inc2_1^2,0 \\
	idle,w''
	\end{array} \right)
	\xrightarrow{v_2-v_1-1,2,!!nc_1,\{4\}} \\
	\left( \begin{array}{c}
	k^2,\param-v_1 -1\\
	inc2_1^1,\param-1 \\
	inc2_1^2,v_2-v_1-1\\
	nc1_1,0
	\end{array} \right)
	\xrightarrow{1,2,!!oc_1,\{4\}} 
	\left( \begin{array}{c}
	k^2,\param-v_1 \\
	idle,\param \\
	inc2_1^2,v_2-v_1\\
	nc2_1,1
	\end{array} \right)
	\xrightarrow{v_1,1,!!tick,\{3,4\}} 
	\left( \begin{array}{c}
	k^1_1,0\\
	idle,\param+v_1\\
	c_2,v_2 \\
	c_1,v_1+1
	\end{array} \right)
	\end{align*}

	\item[Failed increment] $\TwoCMinstruction:incr~c_1~goto~\TwoCMinstruction_1$ (assume there is no process in~$c_2$ for simplicity)
	\begin{align*}
	\scriptscriptstyle
	\left( \begin{array}{c}
	k^1,0 \\
	c_1,v_1 \\
	idle,w
	\end{array} \right)
	\xrightarrow{1,1,!!inc_1,\{2\}} 
	\left( \begin{array}{c}
	k^2,1 \\
	inc1_1^1,v_1+1 \\
	idle,w'
	\end{array} \right)
	\xrightarrow{\param-v_1-4,2,!!nc_1,\{4\}} \\
	\left( \begin{array}{c}
	k^2,\param-v_1 -3\\
	inc2_1^1,\param-3 \\
	nc1_1,0
	\end{array} \right)
	\xrightarrow{3,2,!!oc_1,\{4\}} 
	\left( \begin{array}{c}
	k^2,\param-v_1 \\
	idle,\param \\
	nc1_1,3
	\end{array} \right)
	\xrightarrow{v_1,1,!!tick,\{3,4\}} 
	\left( \begin{array}{c}
	k^1_1,0\\
	idle,\param+v_1\\
	nc1_1,v_1+3
	\end{array} \right)
	\end{align*}
	Notice that, for the rest of the execution, there will be no process encoding~$c_1$, thus no zero test can be achieved.
\end{description}

}
\end{document}